\documentclass[12pt,journal,draftclsnofoot,onecolumn]{IEEEtran} %TWC

\usepackage{amssymb}
\usepackage{amsmath}
\usepackage{amsmath,bm}
\usepackage{amsthm}
\usepackage{graphicx}
\usepackage{subfigure}
\usepackage{cite}
\usepackage{enumerate}
\usepackage{enumitem}
\usepackage[ruled]{algorithm2e}
\usepackage{algorithmic}  
\usepackage{color, soul}
\usepackage[labelsep=period]{caption}
\allowdisplaybreaks[4]
\SetKw{In}{Initialize}
\SetKw{Input}{Input}
\SetKw{Output}{Output}
\newenvironment{indexterms}{
	\begin{center}\small\bfseries Index Terms\end{center}
	\begin{quote}\small
	}
	{\end{quote}}%结束部分定义

\newenvironment{salign}
{
	\begin{equation}
	\setlength{\belowdisplayskip}{-1pt}
}
{\end{equation}}

\begin{document}
	\newtheorem{theorem}{\bf~~Theorem}
	\newtheorem{remark}{\bf~~Remark}
	\newtheorem{definition}{\bf~~Definition}
	\newtheorem{lemma}{\bf~~Lemma}
	\newtheorem{preliminary}{\bf~~Preliminary}
	\newtheorem{proposition}{\it~~Proposition}
	\renewcommand\arraystretch{0.9}
	
	\title{Trajectory Optimization and Resource Allocation for OFDMA UAV Relay Networks}
	\author{\IEEEauthorblockN{
			{Shuhao Zeng}, \IEEEmembership{Student Member, IEEE},
			{Hongliang Zhang}, \IEEEmembership{Member, IEEE},\\
			{Boya Di}, \IEEEmembership{Member, IEEE},
			{and Lingyang Song}, \IEEEmembership{Fellow, IEEE}}\\
		\thanks{S. Zeng, and L. Song are with Department of Electronics, Peking University, Beijing, Peking University, Beijing, 100871, China (email: \{shuhao.zeng, lingyang.song\}@pku.edu.cn).}
		\thanks{H. Zhang is with Department of Electrical Engineering, Princeton University, Princeton, NJ 08544 USA (email: hongliang.zhang92@gmail.com).}
		\thanks{B. Di is with Department of Computing, Imperial College London, London SW7 2AZ, U.K. (email: diboya92@gmail.com).}
	}

	\maketitle
	\begin{abstract}
	In this paper, we consider a single-cell multi-user orthogonal frequency division multiple access~(OFDMA) network with one unmanned aerial vehicle~(UAV), which works as an amplify-and-forward relay to improve the quality-of-service~(QoS) of the user equipments~(UEs) in the cell edge. Aiming to improve the throughput while guaranteeing the user fairness, we jointly optimize the communication mode, subchannel allocation, power allocation, and UAV trajectory, which is an NP-hard problem. To design the UAV trajectory and resource allocation efficiently, we first decompose the problem into three subproblems, i.e., mode selection and subchannel allocation, trajectory optimization, and power allocation, and then solve these subproblems iteratively. Simulation results show that the proposed algorithm outperforms the random algorithm and the cellular scheme.
	\end{abstract}
	\begin{indexterms}
	UAV relay, OFDMA network, trajectory optimization, resource allocation
	\end{indexterms}
	\newpage
	\section{introduction}
%	In wireless systems, relaying has been an effective technique to improve quality-of-service~(Qo\-S)~\cite{TA} and extend coverage~\cite{AEB}. 
	Unmanned aerial vehicle~(UAV) is an emerging facility with a lot of applications, such as environmental monitoring, delivery of goods, etc. Recently, the use of the UAV as relays for providing reliable wireless connections has attracted a lot of attention due to its significant advantages~\cite{XXQWB}. Benefited from the high mobility, the UAV relay can greatly enhance the link quality by dynamically adjusting its location according to the environment, which is superior to conventional static relays~\cite{SHKL}. Besides, the UAV relay is more cost-efficient and can be easily deployed, which makes it more suitable for the emergency cases~\cite{AI_2015}. In addition, the high altitude of the UAV relay can provide a high probability of line-of-sight communication links, serving as a good solution for QoS-critical scenarios.
%	Unmanned aerial vehicle (UAV) is an emerging facility with increasing popularity~\cite{PKA}. According to \cite{GUM}, the global market for commercial UAV applications is estimated at about 2 billion dollars in 2016 and will skyrocket to as much as 127 billion dollars by 2020. There are a wide range of areas where the UAV has been applied, such as military, public and civil applications~\cite{YRT-2016,SHQKL}. Among the various applications, the use of UAV as relays has attracted a lot of attention. 
	
%	The military use of UAVs primarily consists of border surveillance, reconnaissance and strike~\cite{VMG-2018}. Public and civil applications include environmental~\cite{YZKLZ-2017} and natural
%	disaster monitoring~\cite{JRLEJP-2012}, delivery of goods~\cite{AJ-2016}, and so on. Among the various UAV applications, the use of UAV as relays has attracted a lot of attention, due to its significant advantages, such as higher cost-efficiency, faster deployment and higher mobility compared with traditional static relay~\cite{YRT-D}, and the dominance of Line-of-Sight~(LoS) connection. 

In the literature, UAV relaying has been investigated to extend coverage and improve QoS. Existing works on UAV relay networks can be roughly divided into two categories, i.e., distributed UAV relay networks~\cite{RDC_2015,XYNS_2015} and centralized ones~\cite{SHQKL,KWRSS_2016}. In~\cite{RDC_2015}, the authors evaluated the performance of a UAV ad hoc network, where a swarm of UAVs sense the target area and transmit to the ground station directly or through other UAVs, and each UAV relay determines its routing independently. In~\cite{XYNS_2015}, sensors on the ground transmit to a cluster head through UAV relays, and all UAVs take part in the UAV deployment optimization by modeling the deployment problem as a local interaction game, which is solved with an online learning approach. Different from the distributed relaying in the UAV ad hoc network, UAV relaying in cellular networks is operated in a centralized manner, which can be more effective. Most works on cellular UAV networks study a simplified network where the UAV relays for only one user~\cite{SHQKL} or  multiple users over a single carrier~\cite{KWRSS_2016}. In~\cite{SHQKL}, the authors jointly optimized the UAV trajectory and the transmit power of the UAV relay and the mobile device to minimize the outage probability of the network. In~\cite{KWRSS_2016}, the authors considered a wireless sensor network (WSN) with a swarm of UAV relays, where the sensors and the UAVs transmit in a time division multiple access~(TDMA) manner, and the UAV transmission schedule was optimized to improve the system energy efficiency. However, none of these works consider the multi-carrier transmission in orthogonal frequency division multiple access~(OFDMA) manner.
	
The resource allocation or trajectory optimization scheme for a UAV enabled network has also been studied, where a UAV is an OFDMA base station~(BS) to serve multiple user equipments~(UEs)~\cite{MN,YDDLR_2019,RZLDNJJ_2018}. However, unlike the UAV BSs which only consider the links between the UAV and the UEs, the UAV relay considered in this paper should take into account both the links from UEs to the UAV and the link from the UAV to the BS, which makes the resource allocation and trajectory optimization different.

%The resource allocation and trajectory optimization for an OFDMA UAV network where a UAV communicates with multiple downlink users has been carefully studied~\cite{YDDLR_2019,RZLDNJJ_2018}. However, the trajectory optimization and subcarrier allocation method proposed in these papers cannot be directly applied to an OFDMA UAV relay network because the communication through a UAV relay contains two links, both of which will be influenced by the trajectory of the UAV relay.}
	
%	The authors in~\cite{DSD} maximize the energy efficiency of a UAV relay network, which contains one source and one destination, by optimizing the speed and load factor. In~\cite{JYR-2017}, a source transmits to a destination with the help of one UAV relay and the spectrum efficiency as well as energy efficiency of the system is maximized.
	
%	However, in these works, either there is only one source in the network or TDMA is adopted for the transmission. In reality, one network usually contains more than one source.
	
	In this paper, we study a single-cell OFDMA network consisting of one BS, multiple UEs and one UAV relay. To support the transmission from the UEs to the BS, we consider two communication modes, i.e., cellular and relay modes. Specifically, a UE can transmit to the BS directly or through the UAV relay. To improve the throughput while guaranteeing the user fairness, new challenges need to be addressed. Due to the multi-carrier transmission, multiple UEs can transmit to the UAV relay simultaneously. Thus, the trajectory of the UAV relay will influence multiple UEs, which makes the trajectory design more difficult. Moreover, the trajectory, communication mode, and resource allocation are coupled. Therefore, it is not trivial to optimize them jointly.
	
%	The main challenge in this network is that the trajectroy of the UAV relay and the resource allocation are coupled. For the trajectory, the UAV relay tends to provide better channels for the UEs with more transmission resources for the sake of throughput. For the resource allocation, the UEs with better channel conditions will obtain more transmission resources to make full use of the channel. Therefore, it is challenging to jointly optimize the trajectory and the resource allocation. 
	
	To tackle the above challenges, we formulate a series of joint mode selection, subchannel allocation, trajectory optimization, and power allocation problems, each of which is  NP-hard~\cite{M_on}. To solve these problems efficiently, we first decompose each problem into three subproblems, namely the mode selection and subchannel allocation, trajectory optimization, and power allocation problems. Then, we propose a joint mode selection and subchannel allocation, trajectory optimization, and power allocation algorithm~(JMS-T-P) to solve these three subproblems iteratively. Specifically, we reformulate the subchannel allocation and mode selection problem as a many-to-one matching problem, which can be solved by the MC-subchannel matching algorithm~(MSMA). For the trajectory optimization problem, it is further decomposed into horizontal position optimization and altitude optimization problems, where the trajectory optimization~(TO) algorithm is then proposed to solve these two subproblems in a iterative way. Furthermore, the power allocation subproblem is solved by successive convex programming~(SCP) technique.
	
	The contributions of this paper are summarized below,
	\begin{enumerate}[itemindent=0em, label=$\bullet$]
		\item We consider an OFDMA UAV relay network, where each UE can transmit either in the cellular or the relay mode.
		\item To improve system performance, we formulate a series of joint mode selection and subchannel allocation, trajectory optimization, and power allocation problems. To solve these problems efficiently, each problem is decoupled into three subproblems, and an efficient JMS-T-P algorithm is proposed afterwards, where the three subproblems are solved iteratively.
		\item Simulation results show that the proposed JMS-T-P algorithm outperforms the random algorithm and the cellular scheme in terms of sum rate  and user fairness.
	\end{enumerate}
		
	The rest of this paper is organized as follows. In Section \ref{sysmod}, the system model for the UAV relay network is introduced. In Section \ref{PFA}, we formulate a series of joint trajectory optimization and resource allocation problems, each of which is then decomposed into three subproblems. Then, the JMS-T-P algorithm is proposed to solve these three subproblems iteratively in Section \ref{sec_algorithm}, and the convergence and complexity of the JMS-T-P algorithm along with the system performance are discussed in Section \ref{subsection_scoc}. Afterwards, Section \ref{SIM} evaluates the performance of the JMS-T-P algorithm through numerical simulations. Finally, conclusions are drawn in Section~\ref{conclusion}.
	
%	The main contributions of this paper are summarized below,
%	\begin{enumerate}[itemindent=0.5em,label=(\arabic*)]
%		\item We construct an OFDMA cellular network, which contains multiple users and one UAV relay. Then, we formulate a joint mode selection and subchannel allocation, trajectory optimization, and power allocation problem to maximize the system performance.
%		\item An efficient algorithm, i.e., JMS-T-P, is proposed to solve the formulated problem, and then its convergence and complexity is analysed.
%		\item We compare the proposed JMS-T-P with greedy algorithm and static scheme in simulation, and the results show that the proposed algorithm outperforms the other two schemes.
%	\end{enumerate}
	
	\section{System model}\label{sysmod}
	In this section, the scenario is described first, which is followed by channel models and analysis of the UEs' data rates.
	\subsection{Scenario Description}
	As shown in Fig.~\ref{sysmodel}, we consider an uplink OFDMA cellular network with one BS and $N$ UEs, denoted by $\mathcal{N}=\{1,\dots,N\}$. To improve the QoS of cell-edge users, we introduce one UAV as an amplify-and-forward~(AF) relay. 
	
	The transmission timeline is divided into $T$ time slots, denoted by $\mathcal{T}=\{1,\dots,T\}$, where each time slot comprises two phases. To support the transmission from the UEs to the BS, two communication modes are considered in the network:	
	\begin{enumerate}[itemindent=0em, label=$\bullet$]
		\item \textbf{Relay mode}: In the first phase, the UEs transmit to the UAV relay. Then, the UAV relay amplifies the received signal and forwards it to the BS in the second phase~\cite{ZHHXH-2015}.
		\item \textbf{Cellular mode}: The UEs transmit to the BS directly.
	\end{enumerate}
	%	$\bullet$ Relay mode. The transmission contains two phases. In the first phase, the UEs transmit to the UAV relay. The UAV relay amplifies the received signal and forwards it to the BS in the second phase.	
	%	
	%	$\bullet$ Cellular mode. The UEs transmit to the BS directly.\\ 
	Here, we use vector $\bm{\beta}=[\beta_n]_{\mathcal{N}}$ to represent the communication modes of the UEs, where
	\begin{salign}
		\beta_{n}=\left\{
		\begin{aligned}
			&1,~\text{UE $n$ communicates in the relay mode},\\
			&0,~\text{otherwise}.
		\end{aligned}
		\right.
	\end{salign}
	%	Besides, we define $\mathcal{R}=\{n|\beta_{n}=1\}$ as the set of UEs which adopt the relay mode and $\mathcal{D}=\{n|\beta_{n}=0\}$ as the set of UEs which adopt the cellular mode.
	\begin{figure}[!tpb]
		\centering
		\includegraphics[width=4.2in]{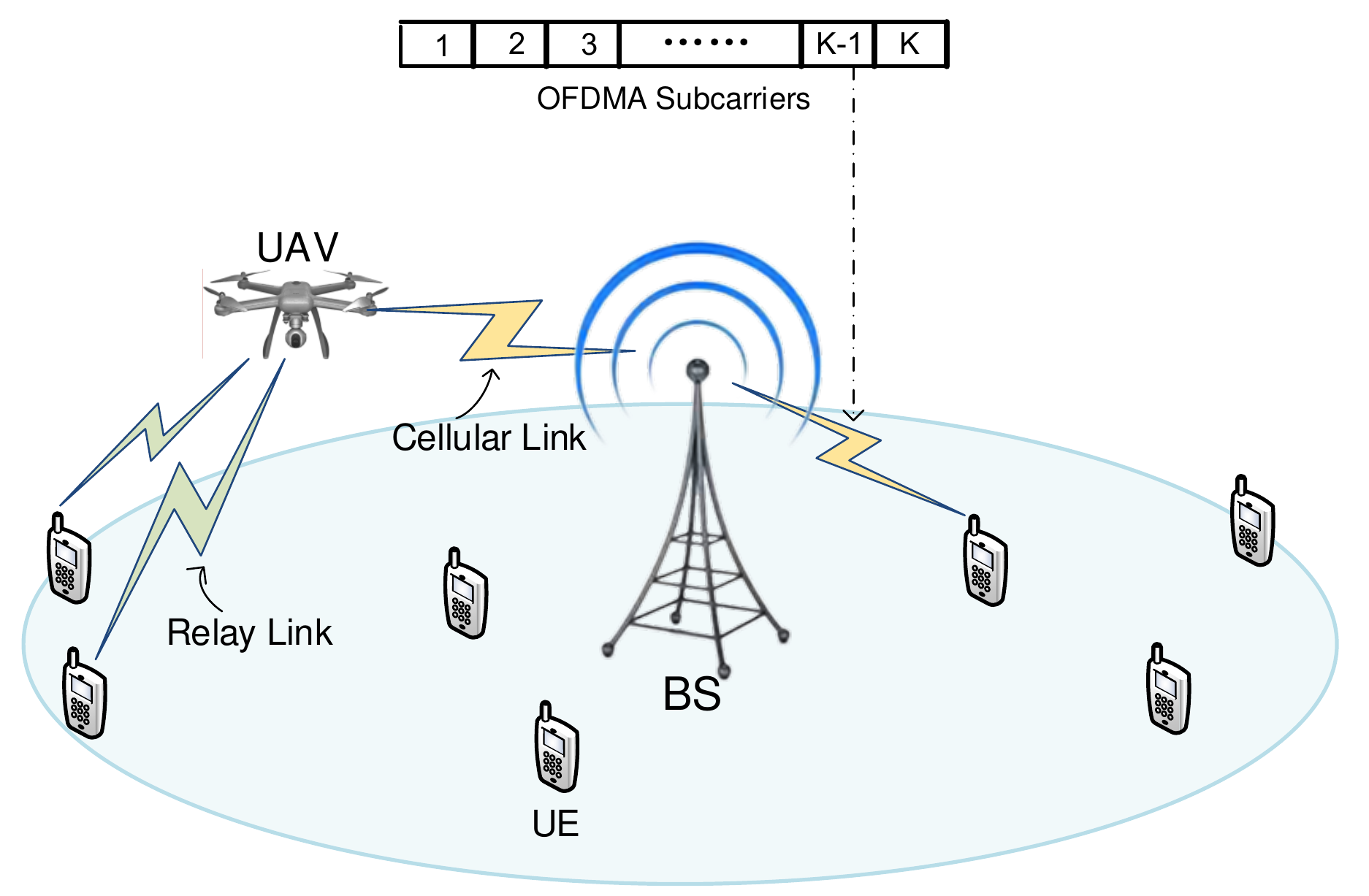}
		\caption{System model of the OFDMA UAV relaying network.}
		\label{sysmodel}
	\end{figure}
	%It is assumed that the cell area is circular and its radius is denoted by $R_0$. We divide the cell into $M+1$ parts. The central part is a circle area with $\frac{2}{3}R_0$ as radius and the users that are located in the area communicate with the BS directly. On the other hand, the users outside the area transmit through UAV relay. The rest of the cell area is divided equally into $M$ parts, which are denoted by the set $\mathcal{M_a}=\{1,\dots,M\}$, with UAV $m$ responsible for users in area $m$. Fig.~$\ref{grouping}$ gives an example of grouping when there are four UAVs.
	
	%%%%%%%%%%%%%%%%%%%%%%%%%%%%%%%%
	%%The set of inner users and outer users are represented by $\mathcal{I} = \{1, \ldots,K\}$ along with $\mathcal{J} = \{K+1, \ldots,N\}$, respectively.
	%%%%%%%%%%%%%%%%%%%%%%%%%%%%%%%%%
	
	 Following the OFDMA manner, the available spectrum of the network is divided into $K$ subchannels, denoted by $\mathcal{K}=\{1,\dots,K\}$. To describe the subchannel allocation to the UEs, we define matrix $\bm{A}=[\alpha_n^{k}]_{\mathcal{N} \times \mathcal{K}}$, where 
	%We denote the subchannel allocation along with the way in which the scheduled UEs communicate with the BS of all time slots as $\bm{A}=[\bm{A[t]}]_{t \in \mathcal{T}}$, where $\bm{A[t]}=[\alpha_n^k[t]]_{\mathcal{N} \times (\mathcal{K} \cup \{0\})}$, and the value of $\alpha_n^k$ is given by
	\begin{equation}
	\alpha_n^{k}=\left\{
	\begin{aligned}
	&1,~\text{subchannel $k$ is assigned to UE $n$},\\
	&0,~\text{otherwise},
	\end{aligned}
	\right.
	\end{equation}
	To avoid severe interference, we assume one subchannel is used by at most one UE, i.e.,
	\begin{equation}
	\label{orth}
	\sum_{n \in \mathcal{N}}{\alpha_{n}^{k}} \le 1, k \in \mathcal{K}.
	\end{equation}
	
	It is assumed that the UEs are static and the BS is located at the center of the cell, whose coordinate is given by $\bm{b}=(0,0,H_B)$, while the UAV is assumed to be mobile, with its locations in the present and former time slots denoted by $\bm{q_U}=(x_U,y_U,z_U)$ and $\bm{q_U^0}=(x_U^0,y_U^0,z_U^0)$, respectively. Define $d_{U,B}$ as the distance between the UAV and the BS. 
	%	Therefore, we have,
	%	\begin{align}
	%	d_{U,B}=\|\bm{q_U}-\bm{b}\|, d_{n,B}=\|\bm{w_n}-\bm{b}\|.
	%	\end{align}
	Due to physical limitations, the flying distance of the  UAV in the present time slot cannot exceed $d_{max}$, i.e.,
	\begin{equation}
	\label{tra1}
	\|\bm{q_U}-\bm{q_U^{0}}\| \le d_{max}.
	\end{equation}
	Besides, we assume that the UAV needs to be higher than the BS to construct a line-of-sight~(LoS) link, i.e.,
	\begin{equation}
	\label{tra2}
	z_U > H_B.
	\end{equation}
	To ensure the normal operation of the UAV relay in the following time slots, the energy consumption of the UAV relay in the present time slot is constrained. According to results in~\cite{CG_2015}, the UAV power consumption mainly comes from flying. Therefore, we have
	\begin{align}
	\label{energy_cons}
	P_f \Delta T \le E_{max},
	\end{align} 
	where $P_f$ is the flying power of the UAV, $\Delta T$ is the length of one time slot, and $E_{max}$ denotes the maximum energy consumption. Here, the maximum energy consumption is determined by deducting the minimum energy consumption required by the following time slots from the remaining battery capacity. To investigate the influence of the energy constraint (\ref{energy_cons}) on the resource allocation and UAV trajectory, we adopt the power consumption model in~\cite{A_2006} to describe the flying power of the UAV, i.e., 
	\begin{align}
	P_f=P_0(1+\frac{3v^2}{U_{tip}^2})+P_i(\sqrt{1+\frac{v^4}{4v_0^4}}-\frac{v^2}{2v_0^2})^{\frac{1}{2}}+\frac{1}{2}d_0\rho sA v^3,
	\end{align} 
	where
	\begin{align}
	P_0=\frac{\delta}{8}\rho sA\Omega^3R^3,\\
	P_i=(1+k)\frac{W^{\frac{3}{2}}}{\sqrt{2\rho A}},
	\end{align}
	$v$ is the speed of the UAV relay, $\delta$ represents profile drag coefficient, $\Omega$ is blade angular velocity, $R$ is rotor radius, $U_{tip}$ denotes tip speed of the rotor blade, $v_0$ is mean rotor induced velocity in hover, $d_0$ represents fuselage drag ratio, $\rho$ is air density, $s$ denotes rotor solidity, $A$ is rotor disc area, $W$ means aircraft weight, and $k$ represents incremental correction factor to induced power.
	
	Denote the transmit power matrices of the UAV as $\bm{P_U}=[P_U^{k}]_{ \mathcal{K}}$, where $P_U^{k}$ is the transmit power of the UAV over subchannel $k$. For convenience, we use $\alpha_n^{k}P^{k}_n$ to represent the transmit power of UE $n$ over subchannel $k$. Besides, the transmit power matrices of the UEs is defined as $\bm{P_M}=[P_n^{k}]_{\mathcal{N} \times \mathcal{K}}$. It is assumed that the transmit powers of each UE and the UAV are constrained, i.e.,
	%$\bm{P^M}=[P^k_n[t]]$ and $\bm{P^U=[P^k_U[t]]}$ to denote the the transmit powers of the UEs and UAV, with $P^k_n[t]$ and $P^k_U[t]$ representing the transmit power of UE $n \in \mathcal{N}$ over subchannel $k$ in time slot $t$ and that of the UAV over subchannel $k$ in time slot $t$, respectively. In each time slot, the total transmit power of an UE or the UAV is constrained, i.e.,
	\begin{align}
	\label{power1}
	\sum_{k \in \mathcal{K}}{\alpha_n^{k}P^{k}_n} \le P_{M}^{max},\\
	\label{power2}
	\sum_{k \in \mathcal{K}}{P^{k}_U} \le P_{U}^{max},
	\end{align}
	where $P_M^{max}$ and $P_{U}^{max}$ are the maximum transmit power of the UE and the UAV, respectively.

	\subsection{Data Transmission}
	In this subsection, we present channel models and analyse the data rates for the cellular and relay modes. 
	\subsubsection{Cellular mode}
 We assume that the channel from UE $n$ to the BS is Rayleigh faded. Therefore, the channel power gain for the channel over subchannel $k$ can be expressed as
 \begin{align}
 h_{n,B}^k=(d_{n,B})^{-\alpha}|g_{n,B}^k|^2
 \end{align}
 where $(d_{n,B})^{-\alpha}$ denotes the pathloss with $\alpha$ being the pathloss coefficient, and $g_{n,B}^k$ represents the Rayleigh small scale fading satisfying $g_{n,B}^k\sim\mathcal{CN}(0,1)$.
%	\begin{equation}
%	\setlength{\abovedisplayskip}{1.5pt} 
%	\setlength{\belowdisplayskip}{1.5pt}
%	\label{ch_UEB}
%	h^k_{n,B}=(d_{n,B})^{-\alpha}\|g^k_{n,B}\|^2,
%	\end{equation}
%\begin{align}
%\label{rec_m}
%P^{R}_{BS,n,k}[t]=P_{n,k}[t](d_{n,BS})^{-\alpha}\|g^k_{n,BS}\|^2.
%\end{align}

Due to the mobility of the UAV relay, inter-subcarrier-interference~(ICI) needs to be considered here. Specifically, in the first phase, the link from UE $n$ in the cellular mode to the BS over subchannel $k$ does not suffer from ICI, since both the UEs in the cellular and relay mode are static. However, due to the UAV mobility, the subcarriers allocated to UEs in the relay mode cause ICI to subchannel $k$ of UE $n$ in the second phase. Therefore, the received signal over subchannel $k$ at the BS in the two phases can be rewritten as,
\begin{align}
\label{rec_cellular_phase_1}
Y_B^{(k,1)}=\sqrt{h_{n,B}^kP_n^k}X_n^{(k,1)}+n_B^{(k,1)},
\end{align}
and
\begin{align}
\label{rec_cellular_phase_2}
Y_B^{(k,2)}=\sqrt{h_{n,B}^kP_n^k}X_n^{(k,2)}+I_B^{k}+n_B^{(k,2)},
\end{align}
respectively, where $X_n^{(k,1)}$ and $X_n^{(k,2)}$ represent data symbols transmitted by UE $n$ over subchannel $k$ in the two phases, respectively, both with unit power, $I_B^{k}$ is the ICI to subchannel $k$ received at the BS in the second phase, and $n_B^{(k,1)}$ and $n_B^{(k,2)}$ are additive white Gaussian noise (AWGN) in the two phases, respectively, which have equal variance $\sigma^2$ and zero mean. Here, we present an approximation for the ICI as shown in Appendix~\ref{ICI}, from which we can find that \emph{the ICI is relatively weak compared with the desired signal, and thus, it can be regarded as a constant for simplicity}. Based on the received signals in (\ref{rec_cellular_phase_1}) and (\ref{rec_cellular_phase_2}), the average data rate of UE $n$ over subchannel $k$ can be expressed as,
\begin{align}
\label{rate_cellular}
R_{n,D}^k=\frac{1}{2}\log_2(1+\frac{P_n^kh^k_{n,B}}{\sigma^2})+\frac{1}{2}\log_2(1+\frac{P_n^kh^k_{n,B}}{\sigma^2+|I_B^{k}|^2}).
\end{align}

To guarantee the QoS, the received SINR of one link over each occupied subchannel cannot be lower than a predetermined threshold. Therefore, we have, 
\begin{equation}
\label{snr_c}
\frac{P_n^kh^k_{n,B}}{\sigma^2}\ge \alpha_n^k(1-\beta_n)\gamma^{th},~
\frac{P_n^kh^k_{n,B}}{\sigma^2+|I_B^k|^2}\ge \alpha_n^k(1-\beta_n)\gamma^{th},
\end{equation}
where $\gamma^{th}$ denotes the threshold for a link from a UE to the BS.
%	Besides, to guarantee QoS, we requires the SNR at the BS from UE $n$ be at least as well as the predetermined threshold $\gamma^{th}$, i.e.,
%	\begin{equation}
%	\setlength{\abovedisplayskip}{2pt} 
%	\setlength{\belowdisplayskip}{2pt}
%	\label{snr_c}
%	\gamma_n^k \ge \alpha_n^k(1-\beta_n)\gamma^{th}.
%	\end{equation}
\subsubsection{Relay mode}
For a UE in the relay mode, the communication to the BS consists of a ground-to-air link from the UE to the UAV relay and an air-to-ground link from the UAV relay to the BS. We assume that both channels are Rician faded, i.e., the channel contains pathloss and small scale Rician fading, which will be elaborated on in the following.

First, we introduce the channel from the UAV relay to the BS. According to the air-to-ground propagation pathloss model~\cite{ASA}, the LoS and NLoS pathloss from the UAV to the BS over subchannel $k$ can be given by
\begin{equation}\label{pathloss}
\begin{aligned}
PL_{LoS}^{k}=L^k  (d_{U,B})^2 \eta_{LoS},\\
PL_{NLoS}^{k}=L^k  (d_{U,B})^2 \eta_{NLoS},
\end{aligned}
\end{equation}
respectively, where $L^k$ is the free space pathloss, and $\eta_{LoS}$ and $\eta_{NLoS}$ are additional attenuation factors due to the LoS and non-line-of-sight~(NLoS) connections. Here, $L^k$ can be expressed as $L^k=(\frac{4 \pi f^k}{c})^2$, where $f^k$ is the frequency of subchannel $k$ in MHz and $c$ is the velocity of light in vacuum. Assume that the antennas on the UAV and the BS are vertically placed. Therefore, the probability of LoS connection is given by
\begin{equation}
\label{lossPro}
PR_{LoS}=\frac{1}{1+a\exp(-b(\theta-a))},
\end{equation}
where $a$ and $b$ are constants depending on the environment, and $\theta=\frac{180}{\pi}\sin^{-1}((z_U-H_B)/d_{U,B})$ is the elevation angle. Based on (\ref{pathloss}) and (\ref{lossPro}), the average pathloss of a link can be expressed as
\begin{equation}
\label{avgloss}
PL_{avg}^{k}=PR_{LoS} PL_{LoS}^{k}+PR_{NLoS} PL_{NLoS}^{k},
\end{equation}
where $PR_{NLoS}=1-PR_{LoS}$ is the probability of the NLoS connection. By combining the pathloss model and small scale fading, the average channel power gain from the UAV relay to the BS over subchannel $k$ can be obtained as
\begin{align}
\label{channelgain}
h_{U,B}^k=\frac{|g_{U,B}^k|^2}{PL_{avg}^k},
\end{align}
where $g_{U,B}^k$ is the Rician small scale fading.

Note that the channel gain from UE $n$ to the UAV relay $h_{n,U}^{k}$ is the same as its reverse channel $h_{U,n}^{k}$ according to the reciprocity of channel~\cite{ASA}. Therefore, we can derive $h_{n,U}^k$ by modeling channel $h_{U,n}^k$ from the UAV relay to UE $n$, which is similar to the channel model from the UAV to the BS. Specifically, the LoS and NLoS pathloss are given by
\begin{align}
PL_{LoS}^{n,k}=L^k (d_{n,U})^2 \eta_{LoS},\\
PL_{NLoS}^{n,k}=L^k (d_{n,U})^2 \eta_{NLoS},
\end{align} 
respectively, and the probability of the LoS and NLoS connections can be expressed as
\begin{align}
PR_{LoS}^n&=\frac{1}{1+a\exp(-b(\theta_n-a))},\\
PR_{NLoS}^n&=1-PR_{LoS}^n,
\end{align}
respectively, where $\theta_n$ is the elevation angle from the UAV to UE $n$. Therefore, the average pathloss is given by 
	\begin{align}
	PL_{avg}^{n,k}=PR_{LoS}^n PL_{LoS}^{n,k}+PR_{NLoS}^n PL_{NLoS}^{n,k},
	\end{align}
	based on which we can obtain the average channel power gain from the UAV relay to UE $n$ as
	\begin{align}
	h_{U,n}^k=\frac{|g_{U,n}^k|^2}{PL_{avg}^{n,k}},
	\end{align}
	where $g_{U,n}^k$ represents the Rician small scale fading. 
%Then, the average received power of BS from UAV $m$ over subchannel $k$ can be expressed as,
%\begin{align}
%\label{PIBsub}
%P^{R}_{BS,m,k}[t]=\frac{P^k_m[t]}{10^{\frac{PL^{avg}_{m,k}[t]}{10}}}.
%\end{align}
%	The received signal at the UAV from UE $n$ over subchannel $k$ is
%	\begin{align}
%	\label{M2U}
%	Y^k_{n,U}=\sqrt{P_n^k h_{n,U}^k}X_n^k+n_U^k,
%	\end{align}
%	where $X_n^k$ is the transmit signal of UE $n$ over subchannel $k$ with unit power and $n_U^k$ is the AWGN at the UAV relay with $\sigma^2$ as variance and 0 as mean. The amplification coefficient of the UAV relay over subchannel $k$ is given by
%	\begin{align}
%	G^k=\frac{P_U^k}{P_n^k h_{n,U}^k+\sigma^2}.
%	\end{align}
%	Therefore, the received signal of BS over subchannel $k$ from the UAV relay can be expressed as
%	\begin{align}
%	\label{U2B}
%	Y^k_{B}=\sqrt{G^k h^k_{U,B}}Y^k_{n,U}+n^k_{B},
%	\end{align}
%	where $n^k_{B}$ is the zero-mean AWGN at the BS with $\sigma^2$ as variance over subchannel $k$. 

Based on the aforementioned channel models, the received signal at the UAV from UE $n$ over subchannel $k$ is\footnote{Since the ICI received at the UAV relay over subchannel $k$ is weak compared with the desired signal, the ICI is neglected here for simplicity.}
\begin{align}
\label{M2U}
Y^k_{n,U}=\sqrt{P_n^k h_{n,U}^k}X_n^k+n_U^k,
\end{align}
where $X_n^k$ is the transmit signal of UE $n$ over subchannel $k$ with unit power, and $n_U^k$ is the AWGN at the UAV relay with $\sigma^2$ as variance and 0 as mean. The amplification coefficient at the UAV relay over subchannel $k$ is given by
\begin{align}
G^k=\frac{P_U^k}{P_n^k h_{n,U}^k+\sigma^2}.
\end{align}
Therefore, the received signal of BS over subchannel $k$ can be expressed as
\begin{align}
\label{U2B}
Y^k_{B}=\sqrt{G^k h^k_{U,B}}Y^k_{n,U}+n^k_{B}+I_B^k,
\end{align}
where $n^k_{B}$ is the AWGN at the BS over subchannel $k$, and $I_B^k$ is ICI to subchannel $k$ from the other subchannels allocated to UEs in the relay mode. Here, we can also find that the ICI is weak compared with the desired signal, and thus, we regard it as a constant for simplicity, as shown in Appendix~\ref{ICI}. According to (\ref{M2U}) and (\ref{U2B}), the SINR at the UAV relay from UE $n$ and that at the BS from the UAV relay can be calculated by
\begin{align}
\gamma^k_{n,U}=\frac{P_n^k h_{n,U}^k}{\sigma^2},
\end{align}
and
\begin{align}
\gamma^k_{U,B}=\frac{P_U^k P_n^k h^k_{U,B} h_{n,U}^k}{\sigma^2\big(P_U^k h^k_{U,B} +(\frac{|I_B^k|^2}{\sigma^2}+1)P_n^k h_{n,U}^k+(\frac{|I_B^k|^2}{\sigma^2}+1)\sigma^2\big)},
\end{align}
respectively. Therefore, the data rate of UE $n$ over subchannel $k$ can be given by
\begin{align}
\label{sub_cap_2}
R^k_{n,R}&=\frac{1}{2}\min\{\log_2(1+\gamma^k_{n,U}),\log_2(1+\gamma^k_{U,B})\}=\frac{1}{2}\log_2(1+\gamma^k_{U,B}),
\end{align}
where $\frac{1}{2}$ follows because two phases are needed for the transmission from UE $n$ to the BS~\cite{JS}. 

Similar to (\ref{snr_c}), to guarantee the QoS for each link, we have
\begin{align}
\label{snr_r_1}
\frac{P_n^kh_{n,U}^k}{\sigma^2}\ge \alpha_n^k\beta_n\gamma_1^{th}, \frac{P_U^kh_{U,B}^k}{\sigma^2+|I_B^k|^2}\ge \alpha_n^k\beta_n\gamma_2^{th},
\end{align}
where $\gamma^{th}_1$ denotes the SINR threshold for a link from UEs to the UAV and $\gamma^{th}_2$ represents that for the link from the UAV to the BS.

%	The total data rate of UE $n$ over all subchannels can be expressed as
%	\begin{align}
%	\label{tot_cap_2}
%	R_{n,R}&=\sum_{k \in \mathcal{K}}{\alpha_n^k R^k_{n,R}}\notag\\
%	&=\sum_{k \in \mathcal{K}}{\frac{1}{2}\alpha_n^k\log_2(1+\frac{P_U^k P_n^k h^k_{U,B} h_{n,U}^k}{\sigma^2(P_n^k h_{n,U}^k +P_U^k h^k_{U,B} +\sigma^2)})}.
%	\end{align}
Define $R_n$ as the sum rate of UE $n$, and we have
\begin{equation}
R_{n}=\beta_{n}\sum_{k \in \mathcal{K}}{\alpha_n^k R^k_{n,R}}+(1-\beta_n)\sum_{k \in \mathcal{K}}\alpha_n^k R_{n,D}^k.
\end{equation}
%Besides, we use $R^{avg}_n$ to represent the average rate of UE $n$ over $T$ time slots, i.e., $R^{avg}_n=\frac{1}{T}\sum\limits_{t \in \mathcal{T}}{R_n^t}$.

%Besides, the data rate of UE $n \in \mathcal{N}$ cannot exceed the capacity of the transmission link, which can be expressed by
%\begin{align}
%\label{ratelim1}
%\alpha_n^0R_n[t] \le \alpha_n^0C^{D}_{n}[t],~\forall n \in \mathcal{N},t \in \mathcal{T},\\
%\label{ratelim2}
%(1-\alpha_n^0)R_n[t] \le (1-\alpha_n^0)C^{R}_{n}[t],~\forall n \in \mathcal{N},t \in \mathcal{T}.
%\end{align}
\section{Problem Formulation and Decomposition}\label{PFA}
In this paper, we aim to improve the sum rate over the UEs while achieving the fairness among these UEs. Therefore, we use the proportional fairness algorithm, and the objective function for one time slot is defined as weighted sum rate of the UEs, i.e., $\sum_{n\in\mathcal{N}}w_nR_n$, where weight $w_n$ is set as the inverse of the average rate over previous time slots{\footnote{To ensure that the objective function makes sense even when the UE never transmit to the BS in previous time slots, we add $0.1$ to the average rate.}}. Then, the joint mode selection and subchannel allocation, trajectory optimization, and power allocation problem for this time slot can be written by
%~\cite{MLSJH,HKYS}
\begin{subequations}\label{overall}
\begin{align}
\label{main1}
&\max_{\bm{\beta}, \bm{A}, \bm{P_{U}}, \bm{P_M}, \bm{q_U}}{\sum_{n\in\mathcal{N}}w_nR_n},\\
\label{main_P}
s.t.~&P_n^{k}, P_U^{k} \ge 0,~\sum_{k \in \mathcal{K}}{P^{k}_U} \le P_{U}^{max},\\
\label{main_PM}
&\sum_{k \in \mathcal{K}}{\alpha_n^{k}P^{k}_n} \le P_{M}^{max},\\
\label{main_S}
&\alpha_n^{k}, \beta_n \in \{0,1\},~\sum_{n \in \mathcal{N}}{\alpha_{n}^{k}} \le 1,\\ 
\label{main_T}
&\|\bm{q_U}-\bm{q_0}\| \le d_{max},~z_U > H_B,\\
\label{energy_flying}
& P_f \Delta T \le E_{max},\\
\label{main_SNR1}
&\frac{P_n^kh^k_{n,B}}{\sigma^2}\ge \alpha_n^k(1-\beta_n)\gamma^{th},~
\frac{P_n^kh^k_{n,B}}{\sigma^2+|I_B^k|^2}\ge \alpha_n^k(1-\beta_n)\gamma^{th},\\
\label{main_SNR2}
&\frac{P_n^kh_{n,U}^k}{\sigma^2}\ge \alpha_n^k\beta_n\gamma_1^{th}, \frac{P_U^kh_{U,B}^k}{\sigma^2+|I_B^k|^2}\ge \alpha_n^k\beta_n\gamma_2^{th}.
%		\label{main12}
%		&\frac{P_n^{k,t}h_{n,B}^{k,t}}{\sigma^2} \ge \alpha_n^{k,t}(1-\beta_n^t)\gamma^{th}, \forall n \in \mathcal{N}, t \in \mathcal{T}, k \in \mathcal{K},\\
%		\label{main13}
%		&\frac{P_U^{k,t}h_{U,B}^{k,t}}{\sigma^2} \ge \alpha_n^{k,t} \beta_n^t \gamma^{th}_{1},\forall n \in \mathcal{N}, t \in \mathcal{T}, k \in \mathcal{K},\\
%		\label{main14}
%		&\frac{P_n^{k,t}h_{n,U}^{k,t}}{\sigma^2} \ge \alpha_n^{k,t}\beta_n^t\gamma^{th}_{2}, \forall n \in \mathcal{N}, t \in \mathcal{T}, k \in \mathcal{K}.
\end{align}
\end{subequations}
Constraint (\ref{main_P}) implies that the transmit power is nonnegative and the transmit power of the UAV is limited. Constraint (\ref{main_PM}) corresponds to (\ref{power1}). The minimum altitude and the maximum moving distance of the UAV relay are given in constraint (\ref{main_T}). Constraint  (\ref{energy_flying}) shows that the energy consumption for flying is constrained. Constraints (\ref{main_SNR1}) and (\ref{main_SNR2}) indicate SINR requirements for each communication link, which correspond to (\ref{snr_c}), and (\ref{snr_r_1}).

To solve the NP-hard problem in (\ref{overall}) efficiently, we decompose it into three subproblems, namely mode selection and subchannel allocation, power allocation, and trajectory optimization problems~\cite{HLY_2017}, and apply block coordinate descent~(BCD) method\footnote{Due to the non-convexity of the original problem, the BCD method will converge to a suboptimal solution.}. Specifically, given transmit power $(\bm{P_U},\bm{P_M})$ and UAV location $\bm{q_U}$, the mode selection and subchannel allocation problem is given by
\begin{equation}
\begin{aligned}
\label{cmsa}
&\max_{\bm{\beta},\bm{A}}{\sum_{n \in \mathcal{N}}{w_nR_n}}\\
s.t.&\text{(\ref{main_PM}), (\ref{main_S}), (\ref{main_SNR1}), and (\ref{main_SNR2})}.	
\end{aligned}
\end{equation}
Given transmit power $(\bm{P_{U}},\bm{P_M})$, communication modes $\bm{\beta}$ and subchannel allocation $\bm{A}$, the location of UAV $\bm{q_U}$ can be optimized by solving the following problem: 
\begin{equation}
\begin{aligned}
\label{subtra}
&\max_{\bm{q_U}}{\sum_{n \in \mathcal{N}}{w_nR_n}}\\
s.t. &\text{(\ref{main_T}), (\ref{energy_flying}), and (\ref{main_SNR2})}.
\end{aligned}
\end{equation} 
Given location $\bm{q_U}$, subchannel allocation $\bm{A}$ and communication modes $\bm{\beta}$, the power allocation problem is given by
\begin{equation}
\begin{aligned}
\label{subpower}		
&\max_{\bm{P_{U}}, \bm{P_M} }{\sum_{n \in \mathcal{N}}{w_nR_n}}\\		
s.t.&~\text{(\ref{main_P}), (\ref{main_PM}), (\ref{main_SNR1}), and (\ref{main_SNR2})}.	
\end{aligned}
\end{equation} 
%\vspace{0.2cm}
\section{Joint Mode Selection and Subchannel Allocation, Trajectory Optimization, and Power Allocation Algorithm Design}\label{sec_algorithm}
In this section, the JMS-T-P algorithm is proposed to solve the subproblems formulated in (\ref{cmsa})-(\ref{subpower}) iteratively.

%is challenging to solve for two reasons. First, subchannel allocation $\bm{A}$ and communication modes $\bm{B}$ are binary. Then, even given $\bm{A}$ and $\bm{B}$, problem (\ref{overall}) is a non-convex optimization problem. 

%	Specifically, for the mode selection and subchannel allocation subproblem, we reformulate it as a many-to-one matching with externalities and propose UE-subchannel matching algorithm~(MSMA) to find a stable matching. For the trajectory optimization subproblem, we further decompose it into horizontal position optimization subproblem and altitude optimization subproblem and solve the two subproblems in an iterative way. The power allocation subproblem is approximated with a convex optimization problem, which can be solved by existing convex optimization techniques such as CVX. Iterations of JMS-T-P will not stop until the objective value of problem (\ref{overall}) converges. 

\subsection{Mode Selection and Subchannel Allocation}
To develop a low complexity algorithm for problem (\ref{cmsa}), we reformulate the subchannel allocation and mode selection as a matching game. Specifically, define $\mathcal{P}$ as the set of UE-communication mode~(MC) pairs, i.e., $\mathcal{P}=\{(n,\beta)|n\in\mathcal{N},\beta\in \{0,1\}\}$. We consider the set $\mathcal{P}$ and the set of subchannels $\mathcal{K}$ as two disjoint sets of selfish and rational players aiming to maximize their own benefits, where the pair $(n,0)$ is considered to be matched with the subchannel $k$ if UE $n$ transmits in the cellular mode with subchannel $k$. Due to constraint (\ref{main_S}), one subchannel can be matched with at most one MC pair. However, one MC pair can be matched with any number of subchannels. Therefore, the matching game is of the many-to-one type. In addition, since each UE can transmit either in the cellular or the relay mode, no subchannels will be matched with $(n,1-\beta)$ if $(n,\beta)$ is matched with certain subchannels. It is assumed that each player has $\emph{preferences}$ over the players in the other set, and the preferences are determined by the utility of the players, which will be described in the following. Suppose that subchannel $k$ is matched with MC pair $(n, \beta)$. Then, the utility of subchannel $k$ over MC pair $(n,\beta)$ is given by
\begin{equation}
\label{utility}
U_S^{(n,\beta),k}=w_n\Big(\beta R_{n,R}^{k}+(1-\beta)R_{n,D}^{k}\Big).
\end{equation}
Similarly, use $\mathcal{K}_0$ to represent the set of subchannels assigned to UE $n$. Therefore, the utility of the MC pair $(n,\beta)$ over $\mathcal{K}_0$, i.e.,  $U_M^{(n,\beta),\mathcal{K}_0}$, is given by
\begin{equation}
\begin{aligned}
U_M^{(n,\beta),\mathcal{K}_0}=\sum_{k \in \mathcal{K}_0}{w_n\Big(\beta R_{n,R}^{k,t}+(1-\beta)R_{n,D}^{k,t}\Big)}\\
\end{aligned}
\end{equation}
Then, based on the utility of the players, the preferences of subchannel $k$, which is denoted by $\succ_k$, can be given below.
\begin{definition}
	For two different MC pairs $(n_1,\beta_1)$ and $(n_2,\beta_2)$, and two matchings $(\Psi_1,\Psi_2)$ where $(n_1,\beta_1)= \Psi_1(k)$ and $(n_2,\beta_2) = \Psi_2(k)$:
	\begin{equation} 
	\label{prefer_sub}
	\setlength{\abovedisplayskip}{1.5pt} 
	\Big((n_1,\beta_1),\Psi_1\Big) \succ_k \Big((n_2,\beta_2),\Psi_2\Big) \Leftrightarrow U_S^{(n_1,\beta_1),k}>U_S^{(n_2,\beta_2),k}. 
	\end{equation}	
\end{definition}
The definition indicates that subchannel $k$ prefers $(n_1,\beta_1)$ in $\Psi_1$ to $(n_2,\beta_2)$ in $\Psi_2$ when the utility it obtains from the MC pair $(n_1,\beta_1)$ is larger than that from $(n_2,\beta_2)$. Similarly, we can define preference $\succ_{(n,\beta)}$ of for the MC pair $(n,\beta)$ over subsets of subchannels, which is neglected here due to the space limit.
\begin{remark}
	The formulated matching game has externalities.
\end{remark} 
\begin{proof}
	Consider the MC pairs $(n,0)$ and $(n,1)$. If the pair $(n,0)$ is matched with some subchannels, $(n,1)$ will not be matched with any subchannel since UE $n$ cannot transmit in both the cellular and relay modes. Therefore, the matching of the pair $(n,0)$ has mutual influence on that of the pair $(n,1)$, which indicates that the formulated game has externalities. 
\end{proof}
Note that under traditional definition of stable matching such as that in~\cite{AM}, the stable matching may not exist for a many-to-one matching game with externalities. Besides, it is computationally hard to find a stable matching if it exists~\cite{ECABA}. Therefore, to solve the formulated matching game, we introduce the notion of pairwise stable matching, and then design a MC-subchannel matching algorithm~(MSMA) to find one such matching, which will be elaborated on in the following. 

Before introducing pairwise stability, we first define swap matching. Specifically, given a matching $\Psi$, the swap matching $\Psi_{k_1}^{k_2}$ is generated from $\Psi$ via exchanging the matches of subchannels $k_1$ and $k_2$ while keeping all other subchannels' matches the same. 
%For convenience, it is assumed that the swap operation is only possible for subchannels which are allocated to UEs with the same communcation mode. ???
It is worthwhile noting that a swap operation may not be approved, due to the interests of the players involved in the swap operation and the constraints of problem (\ref{cmsa}). In the Definition \ref{sbp}, we provide the conditions under which a swap operation will be approved. 
\begin{definition}\label{sbp}
	Given a matching $\Psi$ and a subchannel pair $(k_1,k_2)$. The swap matching $\Psi_{k_1}^{k_2}$ is approved, and $(k_1,k_2)$ is called a swap-blocking pair under $\Psi$ if
	\begin{enumerate}[itemindent=0em, label=$\bullet$]
		\item $\forall v \in \{k_1,k_2,\Psi(k_1),\Psi(k_2)\}$,
		$(\Psi_{k_1}^{k_2}(v),\Psi_{k_1}^{k_2}) \succeq_v (\Psi(v),\Psi)$, 
		\item $\exists v \in \{\Psi(k_1),\Psi(k_2)\}$,
		$(\Psi_{k_1}^{k_2}(v),\Psi_{k_1}^{k_2}) \succ_v (\Psi(v),\Psi)$,
		\item The constraint of problem (\ref{cmsa}) hold for $\Psi_{k_1}^{k_2}$.
	\end{enumerate}
\end{definition}
The first item in Definition~\ref{sbp} indicates that the utilities of subchannels $\{k_1,k_2\}$ and MC pairs $\{\Psi(k_1),\Psi(k_2)\}$ do not decease after the exchange. The second item means 
at least one MC pair in $\{\Psi(k_1),\Psi(k_2)\}$ obtains a larger utility after the exchange. The last item states that the transmit power and SNR constraints of problem (\ref{cmsa}) are satisfied after the exchange. Based on the definition of swap-blocking pair, we say a matching $\Psi$ is pairwise stable if it is not blocked by any swap-blocking pair under $\Psi$. 

To find a pairwise stable matching, we propose algorithm MSMA. First, we present the following lemma on the relationship between the subchannel allocation and the communication mode, which will be utilized in the MSMA algorithm.
\begin{lemma}\label{lemma_sm}
	Suppose that the subchannel allocation $\bm{A}^*$ and the communication mode $\bm{\beta}^*$ are optimal, and
	\begin{align}\label{SNR}
	\frac{P_n^kh^k_{n,B}}{\sigma^2+|I_B^k|^2} \ge (\alpha_n^{k})^*\gamma^{th},~
	\frac{P_n^kh_{n,U}^k}{\sigma^2} \ge (\alpha_n^{k})^* \gamma^{th}_{1},~
	\frac{P_U^kh_{U,B}^k}{\sigma^2+|I_B^k|^2} \ge (\alpha_n^{k})^*\gamma^{th}_{2}, \forall k.
	\end{align}
	Then, we have
	\begin{enumerate}
		\item \text{$(\beta_n)^*$=1, if $\sum\limits_{k \in \mathcal{K}}{(\alpha_n^{k})^*R_{n,R}^{k}}>\sum\limits_{k \in \mathcal{K}}{(\alpha_n^{k})^*R_{n,D}^{k}}$.}
		\item \text{$(\beta_n)^*$=0, if $\sum\limits_{k \in \mathcal{K}}{(\alpha_n^{k,t})^*R_{n,R}^{k}}<\sum\limits_{k \in \mathcal{K}}{(\alpha_n^{k,t})^*R_{n,D}^{k}}$.}
	\end{enumerate}
\end{lemma}
\begin{proof}
	See Appendix~\ref{append_label_sm}.
\end{proof}
The MSMA algorithm is composed of initialization step and swap step. In the initialization step, Lemma~\ref{lemma_sm} is utilized to help determining the communication modes, which can reduce the swap operations in the swap step. Then, in the swap step, we keep searching for approved swap matching\footnote{The MSMA algorithm is performed by the BS. After obtaining the positions of the UEs, the BS performs the matching algorithm.}. The details of MSMA is provided in Algorithm~\ref{MSMA_1}.

\begin{algorithm}
	\caption{MC-SUBCHANNEL MATCHING ALGORITHM~(MSMA)}
	\label{MSMA_1}
	\Input{$\bm{q_U}$, $\bm{P_U}$, $\bm{P_M}$}\\
	\Output{$\bm{\beta}$,$\bm{A}$}\\
	\In{$\bm{\beta}^0$,$\bm{A}^0$}\\
	\Repeat{no swap operation is performed in the current iteration}
	{
		\For{every subchannel $k \in \mathcal{K}$}{
			\While{true}{
			The BS searches $\mathcal{K}\backslash\{k\}$ for $k'$ such that swap matching $\Psi_k^{k'}$ is approved, and the swap operation involving $(k,k')$ as well as $(\Psi(k),\Psi(k'))$ is never executed in the current round;\\
			\uIf{any $k'$ is found}
				{Update $\Psi$ with the found approved swap matching $\Psi_k^{k'}$;}
			\Else{Break;}
			}	
		}
	}
\end{algorithm}

\subsection{Trajectory Optimization}
\label{Trajectory_Optimization}	
Problem (\ref{subtra}) is still non-convex due to the non-concave objective function and non-convex constraints (\ref{energy_flying}) and (\ref{main_SNR2}). To tackle it efficiently, we further decompose it into two subproblems, i.e., horizontal position and altitude optimization subproblems, and then propose a trajectory optimization algorithm~(TO) to solve problem (\ref{subtra}), where the two subproblems are solved iteratively. In the following, we will first solve the two subproblems and then elaborate on the TO algorithm. 
\subsubsection{Horizontal position optimization}
\label{HPO}
Given the altitude $z_U$, the horizontal position optimization subproblem is given by
\begin{align}\label{subtra_hor}
\begin{split}
&\max_{x_U,y_U}{\sum_{n \in \mathcal{N}}{w_nR_n}},\\
s.t.~&\|\bm{q_U}-\bm{q_0}\| \le d_{max},\\
	 &\text{(\ref{energy_flying}) and (\ref{main_SNR2})}.
\end{split}
\end{align}
%\begin{theorem}
%	The feasible set of problem (\ref{subtra_alt}) is a closed convex set.
%\end{theorem}
%\begin{proof}
%	See Appendix~\ref{clo_con}
%\end{proof} 
which is still non-convex due to constraints (\ref{main_SNR2}), (\ref{energy_flying}), and the non-concave objective function. To solve this problem efficiently, we can apply SCP, which will be described in detail in the following.  

In iteration $j$, in order to find convex approximations for the non-convex functions, we first introduce concave approximations for the channel gains $h_{U,B}^{k}$ and $h_{n,U}^{k}$. Specifically, $h_{U,B}^{k}$ can be approximated as
	\begin{align}
	H_{U,B}^{k}=|g_{U,B}^k|^2\left(\frac{2}{L^{k}(x_U^{j-1},y_U^{j-1})}-\frac{1}{\Big(L^{k}(x_U^{j-1},y_U^{j-1})\Big)^2}L^{k}\right),
	\end{align}
	where $L^{k}=
	\overline{PR_{LoS}} PL_{LoS}^{k}+(1-\overline{PR_{LoS}}) PL_{NLoS}^{k}$,
	$\overline{PR_{LoS}}=
	\frac{2}{D^{j-1}}-\frac{1}{(D^{j-1})^2}-\frac{a}{(D^{j-1})^2}\exp\Big(-b(\overline{\theta}-a)\Big)$,
	$D^{j-1}=1+a\exp\bigg(-b\Big(\overline{\theta}(x_U^{j-1},y_U^{j-1})-a\Big)\bigg)$,
	$\overline{\theta}=
	\frac{180}{\pi}\Big(\sin^{-1}(\frac{1}{C^{j-1}})+(\frac{-1}{C^{j-1}\sqrt{(C^{j-1})^2-1}})\newline(\frac{d_{U,B}}{z_U-H_B}-C^{j-1})\Big)$, and
	$C^{j-1}=\frac{d_{U,B}(x_U^{j-1},y_U^{j-1})}{z_U-H_B}$. Also, $H_{n,U}^{k}$, the approximation for $h_{n,U}^{k}$, is defined as
\begin{align}
H_{n,U}^{k}=|g_{U,n}^k|^2\left(\frac{2}{M_n^{k}(x_U^{j-1},y_U^{j-1})}-\frac{1}{(M_n^{k}(x_U^{j-1},y_U^{j-1}))^2}M_n^{k}\right),
\end{align}
where $M_n^{k}=\overline{PR_{LoS}} PL_{LoS}^{k}+(1-\overline{PR_{LoS}}) PL_{NLoS}^{k}$, $\overline{PR_{LoS}^{n}}=\frac{2}{B_{n}^{j-1}}-\frac{1}{(B_{n}^{j-1})^2}-\frac{a}{(B_{n}^{j-1})^2}\exp\Big(-b(\overline{\theta_n}-a)\Big)$, $B_{n}^{j-1}=1+a\exp\bigg(-b\Big(\overline{\theta_n}(x_U^{j-1},y_U^{j-1})-a\Big)\bigg)$, $\overline{\theta_n}=\frac{180}{\pi}\Big(\sin^{-1}(\frac{1}{A_{n}^{j-1}})+(\frac{-1}{A_{n}^{j-1}\sqrt{(A_{n}^{j-1})^2-1}})\newline(\frac{d_{n,U}}{z_U}-A_{n}^{j-1})\Big)$, and $A_{n}^{j-1}=\frac{d_{n,U}(x_U^{j-1},y_U^{j-1})}{z_U}$. Then, we have the following theorem:
\begin{theorem}
	\label{con_low_approx}
	$H_{U,B}^{k}$ and $H_{n,U}^{k}$ are concave lower bounds for $h_{U,B}^{k}$ and $h_{n,U}^{k}$, respectively.
\end{theorem}
\begin{proof}
	The Proof is shown in Appendix~\ref{Append_2}.
\end{proof}

Using the above approximations, we can mitigate the non-convexity of constraints (\ref{main_SNR2}) brought by the channel gains, i.e., 
%Specifically, we first replace the channel gains in (\ref{main_SNR2}) with its approximations. Although this replacement transform the left-hand-side~(LHS) of (\ref{main_SNR2}) into convex functions, this constraint is still not convex. Therefore, we further approximate the LHS with its the first-order Taylor expansion at a given local point, i.e.,
\begin{align}
\label{snr_approx}
\frac{P_n^kH_{n,U}^k}{\sigma^2}\ge \alpha_n^k\beta_n\gamma_1^{th}, \frac{P_U^kH_{U,B}^k}{\sigma^2+|I_B^k|^2}\ge \alpha_n^k\beta_n\gamma_2^{th}.
%\bm{Ta}(\frac{P_U^{k,t}H_{U,B}^{k,t}}{\sigma^2})\ge \alpha_n^{k,t}\beta_n^t\gamma_1^{th},~\bm{Ta}(\frac{P_n^{k,t}H_{n,U}^{k,t}}{\sigma^2})\ge \alpha_n^{k,t}\beta_n^t\gamma_2^{th},
\end{align}

Besides, we can also transform the objective function of problem (\ref{subtra_hor}) into a concave one with the approximations $H_{U,B}^{k}$ and $H_{n,U}^{k}$, where the main idea is to find a concave approximation for $R_{n,R}^{k,t}$. The transformation process contains three steps: 

\textbf{Step 1}: Replace $h_{n,U}^{k}$ and $h_{U,B}^{k}$ with $H_{n,U}^{k}$ and $H_{U,B}^{k}$, respectively. As a result, $R_{n,R}^{k}$ becomes the difference of two concave functions~(DC), i.e.,  $\hat{J}_n^{k}$ and $\hat{I}_n^{k}$, where
	%First, we rewrite the objective function as
	%\begin{equation}
	%%\setlength{\abovedisplayskip}{1.5pt} 
	%%\setlength{\belowdisplayskip}{1.5pt}
	%\label{open}
	%\sum_{n \in \mathcal{N}}{\log (R^{avg}_n+1)}=\sum_{n \in \mathcal{N}}{\log \bigg(\frac{1}{T} \sum_{t \in \mathcal{T}}\sum_{k \in \mathcal{K}}\Big(\alpha_n^{k,t}\beta_{n}^t R_{n,R}^{k,t}+\alpha_n^{k,t}(1-\beta_n^t)R_{n,D}^{k,t}\Big)+1\bigg)}.
	%\end{equation}
	%To find a concave approximation for $R_{n,R}^{k,t}$, we replace $h_{n,U}^{k,t}$ and $h_{U,B}^{k,t}$ with their concave approximations. As a result, $R_{n,R}^{k,t}$ becomes the  difference of concave functions~(DC) $\hat{J}_n^{k,t}$ and $\hat{I}_n^{k,t}$, where
	\begin{equation}
	\begin{aligned}
	%	R_{n,R}^{k,t}=&\frac{1}{2}\Big(\log(P_n^{k,t} h_{n,U}^{k,t}+\sigma^2)+\log(P_U^{k,t} h^{k,t}_{U,B}+\sigma^2)\Big)-\frac{1}{2}\log\Big(\sigma^2(P_n^{k,t} h_{n,U}^{k,t} +P_U^{k,t} h^{k,t}_{U,B} +\sigma^2)\Big)\notag \\
	\hat{J}_n^{k}&=\frac{1}{2}\bigg(\log(P_n^{k} H_{n,U}^{k}+\sigma^2)+\log(P_U^{k} H_{U,B}^{k}+(\frac{|I_B^k|^2}{\sigma^2}+1)\sigma^2)\bigg),\\
	\hat{I}_n^{k}&=\frac{1}{2}\log\Big(\sigma^2(P_n^k H_{n,U}^k +(1+\frac{|I_B^k|^2}{\sigma^2})P_U^k H^k_{U,B} +(1+\frac{|I_B^k|^2}{\sigma^2})\sigma^2)\Big).
	\end{aligned}
	\end{equation}
	
	%	\begin{align}
	%	%	R_{n,R}^{k,t}=&\frac{1}{2}\Big(\log(P_n^{k,t} h_{n,U}^{k,t}+\sigma^2)+\log(P_U^{k,t} h^{k,t}_{U,B}+\sigma^2)\Big)-\frac{1}{2}\log\Big(\sigma^2(P_n^{k,t} h_{n,U}^{k,t} +P_U^{k,t} h^{k,t}_{U,B} +\sigma^2)\Big)\notag \\
	%	R_{n,R}^{k,t}\ge&\hat{J}_n^{k,t}-\hat{I}_n^{k,t}\notag\\
	%	=&\frac{1}{2}\bigg(\log(P_n^{k,t} H_{n,U}^{k,t}+\sigma^2)+\log(P_U^{k,t} H_{U,B}^{k,t}+\sigma^2)\bigg)\notag\\
	%	&-\frac{1}{2}\log\Big(\sigma^2(P_n^{k,t}H_{n,U}^{k,t} +P_U^{k,t} H_{U,B}^{k,t} +\sigma^2)\Big)
	%	\end{align}	
	\textbf{Step 2}: Approximate $\hat{I}_n^{k}$ by its first-order Taylor expansion. Then, we can obtain a concave approximation for $R_{n,R}^{k}$, which is given by 
	\begin{align}
	\hat{R}_{n,R}^{k}=\hat{J}_n^{k}-\bm{Ta}(\hat{I}_n^{k}),
	\end{align}
	where $\bm{Ta}(.)$ represents the first-order Taylor expansion at a given point $(x_U^{j-1},y_U^{j-1})$ with respect to $(x_U,y_U)$.
	
	\textbf{Step 3}: Replace $R_{n,R}^{k}$ in the expression of objective function with $\hat{R}_{n,R}^{k}$, and the objective function is transformed into a concave function.

To ensure the convexity of the energy constraint~(\ref{energy_flying}), we find an upper bound for the power consumption in the present time slot as follows,
\begin{align}
P_f\le P_0(1+\frac{3v^2}{U_{tip}^2})+P_i+\frac{1}{2}d_0\rho sA v^3 \triangleq P_u,
\end{align}
and the energy consumption constraint can be replaced with,
\begin{align}
\label{battery_up}
 P_u \Delta T\le E_{max},
\end{align}
which is a convex constraint.

Therefore, the convex optimization subproblem for iteration $j$ can be written as,
	\begin{align}
	\begin{split}
	\label{subtra_hor_v3}
	&\max_{x_U,y_U}{\sum_{n \in \mathcal{N}}{\bigg(w_n\sum_{k \in \mathcal{K}}\Big(\alpha_n^{k}\beta_{n} \hat{R}_{n,R}^{k}+\alpha_n^{k}(1-\beta_n)R_{n,D}^{k}\Big)\bigg)},}\\
	s.t.~&\|\bm{q_U}-\bm{q_0}\| \le d_{max},\\
	&\text{(\ref{snr_approx}), and (\ref{battery_up})},
	\end{split}
	\end{align}
which can be solved by standard convex optimization tools~\cite{SL}.

%\begin{theorem}
%	\label{theo_lb_hp}
%	The optimal value of problem (\ref{subtra_hor_v3}) is a lower bound for that of the horizontal position optimization problem (\ref{subtra_hor}).
%\end{theorem}
%\begin{proof}
%	The Proof is shown in Appendix \ref{pro_lb_hp}.
%\end{proof}

\subsubsection{Altitude optimization}
Given the horizontal position, the altitude optimization subproblem is given by
\begin{equation}\label{subtra_alt_t}
\begin{aligned}
&\max_{z_U}{\sum_{n \in \mathcal{N}}{w_nR_n}}\\
s.t.~&\text{(\ref{main_T}), (\ref{energy_flying}), and (\ref{main_SNR2})}.
%	\|\bm{q_U^{t+1}}-\bm{q_U^{t}}\| \le d_{max}, t=1, \dots, T-1,\\
%	\label{cons_subtra2_alt}
%	&z_U^t > H_B, t=2,\dots,T
\end{aligned}
\end{equation}
Assume that the maximum moving distance $d_{max}$ is much smaller than the altitude of the UAV relay, and we have the following theorem
%\begin{theorem}
%	$h_{U,B}^{k,t}$ can be approximated by
%	\begin{align}
%	h_{U,B}^{k,t}\approx\frac{\|g_{U,B}^{k,t}\|^2}{E^{k,t}z_U^t+F^{k,t}},
%	\end{align}
%	where
%	\begin{align}
%	\begin{split}
%	F^{k,t}=PR_{LoS}^t(z_U^{t-1}),\\???
%	E^{k,t}=\Big(PL_{LoS}^{k,t}(z_U^{t-1})-PL_{NLoS}^{k,t}(z_U^{t-1})\Big)PR_{NLoS}^t(z_U^{t-1})PR_{LoS}^t(z_U^{t-1})\frac{\frac{180}{\pi}b}{\sqrt{1-sin^2(\frac{\pi}{180}\theta^t(z_U^{t-1}))}},\\
%	\end{split}
%	\end{align} 
%	Besides, $h_{n,U}^{k,t}$ can be approximated by 
%	\begin{align}
%	h_{n,U}^{k,t}\approx\frac{\|g_{U,n}^{k,t}\|^2}{M_n^{k,t}z_U^t+N_n^{k,t}},
%	\end{align}
%	where $N_n^{k,t}=PR_{LoS}^{n,t}(z_U^{t-1})$ and $M_n^{k,t}=L^k(\eta_{LoS}-\eta_{NLoS})d_{n,U}^t(z_U^{t-1})(1-PR_{LoS}^{n,t}(z_U^{t-1}))PR_{LoS}^{n,t}(z_U^{t-1})\frac{\frac{180}{\pi}b}{\sqrt{1-sin^2(\frac{\pi}{180}\theta_n^t(z_U^{t-1}))}}$.
%\end{theorem}
\begin{theorem}
	\label{the_PP}
	Suppose $z_U$ is feasible for problem (\ref{subtra_alt_t}). $PR_{LoS}$ in $h_{U,B}^{k}$ can be approximated by a linear function, i.e,
	\begin{align}
	PR_{LoS}= E+F\frac{z_U-z_0}{d_{U,B}^t(z_0)}\triangleq  \mathbb{PR}_{LoS},
	\end{align}
	where $F=PR_{NLoS}(z_0)PR_{LoS}(z_0)\frac{\frac{180}{\pi}b}{\sqrt{1-\sin^2(\frac{\pi}{180}\theta^t(z_0))}}$ and $E=PR_{LoS}(z_0)$. Besides, $PR_{LoS}^{n}$ can also be approximated by a linear function, i.e.,
	\begin{align}
	PR_{LoS}^{n} = S_n+N_n\frac{z_U-z_0}{d_{n,U}(z_0)}\triangleq \mathbb{PR}_{LoS}^{n},
	\end{align}
	where $N_n=PR_{NLoS}^{n}(z_0)PR_{LoS}^{n}(z_0)\frac{\frac{180}{\pi}b}{\sqrt{1-\sin^2(\frac{\pi}{180}\theta_n(z_0))}}$ and $S_n=PR_{LoS}^{n}(z_0)$.
\end{theorem}
\begin{proof}
	See Appendix \ref{Append_3}.
\end{proof}
Besides, based on the assumption, when $z_U$ is a feasible point for problem (\ref{subtra_alt_t}), the distance $d_{n,U}$ and $d_{U,B}$ can be approximated by $d_{n,U}(z_0)$ and $d_{U,B}(z_0)$, respectively. Therefore, we have
	\begin{align}
	h_{U,B}^{k} \approx \frac{|g_{U,B}^k|^2}{\mathbb{PR}_{LoS} (L^k (d_{U,B}(z_0))^2 \eta_{LoS})+(1-\mathbb{PR}_{LoS})(L^k (d_{U,B}(z_0))^2 \eta_{NLoS})} \triangleq (h_{U,B}^{k})',\\
	h_{n,U}^{k} \approx \frac{|g_{U,n}^k|^2}{\mathbb{PR}_{LoS}^{n}(L^k (d_{n,U}(z_0))^2 \eta_{LoS})+(1-\mathbb{PR}_{LoS}^{n})(L^k (d_{n,U}(z_0))^2 \eta_{NLoS})} \triangleq (h_{n,U}^{k})',
	\end{align}
	where $(h_{U,B}^{k})'$ and $(h_{n,U}^{k})'$ are convex. Thus, we can use SCP method to solve problem (\ref{subtra_alt_t}). Specifically, in the $l$-th iteration, we first obtain concave lower bounds of $(h_{U,B}^{k})'$ and $(h_{n,U}^{k})'$ by applying the first-order Taylor expansions at $(z_U)^{l-1}$. Then, convex constraints and a lower bound for the objective function can be obtained by the steps introduced in Section~\ref{HPO}. Besides, constraint (\ref{energy_flying}) can also be replaced with the constraint (\ref{battery_up}), which is convex with respect to UAV altitude $z_U$.
%To obtain a locally optimal solution of (\ref{subtra_alt}), a gradient-based iterative method is proposed. Specifically, in each iteration, the altitude $\bm{z}_U$ is updated in the direction of the gradient. To gaurantee the convergence of the altitude optimization method, the objective value needs to be non-increasing after each iteration. Besides, it is required that the altitude is always feasible as iterations go on. To satisfy the above two requirements, we choose the step size in the following way. For the step size in iteration $l$, denoted by $\delta^l$, we first choose a random value for it. Then, we halve the step size $\delta^l$ until the objective value under $\bm{z}_U^l+\delta^l$ is larger than that under $\bm{z}_U^l$ and $\bm{z}_U^l+\delta^l$ is feasible.

%Based on the above two subsections, algorithm TO is proposed to solve the trajectory optimization problem (\ref{subtra}). In iteration $j$ of algorithm TO, we first optimize the horizontal position, where we set the horizontal position in the $(j-1)$-th iteration as $(\bm{x}_{loc},\bm{y}_{loc})$. Then, given the horizontal position result, the altitude is optimized, where the altitude in the $(j-1)$-th iteration is selected as the starting point for the altitude optimization, denoted by $\bm{z}_U^{int}$. Besides, in the initialization step, TO algorithm is initialized with the most recently updated UAV trajectory in the outer iteration. The algorithm is summarized in Algorithm~\ref{alg_TO}.

Based on the above two subsections, TO algorithm is proposed to solve the trajectory optimization problem (\ref{subtra}). In iteration $s$ of TO algorithm, we first optimize the horizontal position. Then, given the horizontal position result, the altitude is optimized. The algorithm is summarized in Algorithm~\ref{alg_TO}.
%\begin{algorithm}
%	\caption{TO algorithm}
%	\label{alg_TO}
%	\Input{$\bm{A}$, $\bm{B}$, $\bm{P_U}$, $\bm{P_M}$.}\\
%	\Output{$\bm{Q}$.}\\
%	\In{$\bm{Q}^0,j=1$.}\\
%	\Repeat{
%		\text{The fractional increase of the objective value is below the pre-determined threshold $\epsilon_{TO}>0$}
%	}
%	{
%		Solve problem (\ref{subtra_hor}) given $\bm{z}_U^{j-1}$, where $(\bm{x}_{loc},\bm{y}_{loc})=(\bm{x}_U^{j-1},\bm{y}_U^{j-1})$. Denote the solution as $(\bm{x}_U^j,\bm{y}_U^j)$.\\
%		Solve problem (\ref{subtra_alt}) given $(\bm{x}_U^j,\bm{y}_U^j)$, where $\bm{z}_U^{int}=\bm{z}_U^{j-1}$. Denote the solution as $\bm{z}_U^{j}$.\\
%		Update $j=j+1$.
%	}
%\end{algorithm}
\begin{algorithm}
	\caption{TO Algorithm}
	\label{alg_TO}
	\Input{$\bm{A}$, $\bm{\beta}$, $\bm{P_U}$, $\bm{P_M}$.}\\
	\Output{$\bm{q_U}$.}\\
	\In{$\bm{q_U}^0,s=1$.}\\
	\Repeat{
		\text{The fractional increase of the objective value is below the pre-determined threshold $\epsilon_{TO}>0$.}
	}
	{
		Solve problem (\ref{subtra_hor}) given $z_U^{s-1}$, and denote the solution as $(x_U^s,y_U^s)$.\\
		Solve problem (\ref{subtra_alt_t}) given $(x_U^s,y_U^s)$, and denote the solution as $z_U^{s}$.\\
		Update $s=s+1$.
	}
\end{algorithm}
\subsection{Power Allocation}
In this subsection, we optimize the power allocation by solving problem (\ref{subpower}). For problem (\ref{subpower}), the constraints are convex while the objective function is non-concave. We also use the SCP method to solve this problem. For iteration $o$, we first transform the objective function of power allocation problem (\ref{subpower}) into a concave one, and then solve the modified power allocation problem, which will be described in detail in the following. Note that $R_n$ is a DC, i.e., $R_n=K_n-M_n$, where
\begin{align}
K_n=\sum_{k \in \mathcal{K}}{\frac{1}{2}\beta_n\alpha_n^{k}\log_2\Big((P_n^{k} h_{n,U}^{k}+\sigma^2)(P_U^{k} h_{U,B}^{k}+(\frac{|I_B^k|^2}{\sigma^2}+1)\sigma^2)\Big)}+\sum_{k \in \mathcal{K}}{(1-\beta_n)\alpha_n^{k}R_{n,D}^{k}},
\end{align}
\begin{align}
M_n=\sum_{k \in \mathcal{K}}{\frac{1}{2}\beta_n\alpha_n^{k}\log_2\Big(\sigma^2(P_n^k h_{n,U}^k +(1+\frac{|I_B^k|^2}{\sigma^2})P_U^k h^k_{U,B} +(1+\frac{|I_B^k|^2}{\sigma^2})\sigma^2)\Big)}.
\end{align}
Therefore, by replacing $M_n$ with its first-order Taylor expansion, we can obtain a concave approximation for $R_n$, i.e., $K_n-\bm{Tp}(M_n)$,
where $\bm{Tp}(.)$ represents the first-order Taylor expansion at $(\bm{P_U}^{o-1},\bm{P_M}^{o-1})$. Therefore, the power allocation problem transforms into
\begin{subequations}\label{subpower_low}
	\begin{align}
	\label{main_subpower_low}
	&\max_{\bm{P_{U}}, \bm{P_M} }{\sum_{n \in \mathcal{N}}{ w_n\Big(K_n-\bm{Tp}(M_n)\Big)}}\\
	s.t.&~\text{(\ref{main_P}), (\ref{main_PM}), (\ref{main_SNR1}), and (\ref{main_SNR2})},
	\end{align}
\end{subequations}
which is convex since the non-concavity of the objective function has been removed.
%Due to the concavity of $M_n^t$, we have $\bm{Tp}(M_n^t) \ge M_n^t$. Therefore, the approximated objective function is a lower bound for the original function. As a result, the optimal value of problem (\ref{subpower_low}) is no larger than that of the original problem problem (\ref{subpower}).
%	\begin{align}
%	\label{power_low}
%	K_n^t-M_n^t \ge K_n^t-(M_n^t)^{r-1}-\Big(\nabla M_n^t \big|_{(\bm{P_U}^{r-1},\bm{P_M}^{r-1})}\Big)^T(\Delta \bm{P_U},\Delta \bm{P_M}).
%	\end{align}
%	where $(M_n^t)^{r-1}$ denotes the value of $M_n^t$ under $(\bm{P_U}^{r-1},\bm{P_M}^{r-1})$ and $(\Delta \bm{P_U},\Delta \bm{P_M})=(\bm{P_U}-\bm{P_U}^{r-1},\bm{P_M}-\bm{P_M}^{r-1})$. 

%	With the lower bounds in (\ref{power_low}) and the given local point $(\bm{P_U}^{r-1},\bm{P_M}^{r-1})$, 
%	With the approximation for $R_n^t$, an approximation for problem (\ref{subpower}) can be obtained,
%	\begin{subequations}\label{subpower_low}
%	\begin{align}
%	\label{main_subpower_low}
%		&\max_{\bm{P_{U}}, \bm{P_M} }{\sum_{n \in \mathcal{N}}{\log \Bigg(\frac{1}{T}\sum_{t \in \mathcal{T}}{\bigg(K_n^t-T_1^{r-1}(M_n^t)\bigg)}+1\Bigg)}}\\
%	s.t.&\text{(\ref{main5})-(\ref{main10}) and  (\ref{main12})-(\ref{main14})}.
%	\end{align}
%	\end{subequations}
%	Due to the concavity of $(K_n^t-T_1^{r-1}(M_n^t))$, the objective function (\ref{main_subpower_low}) is concave. Besides, all constraints of the approximation (\ref{subpower_low}) are convex. Therefore, problem (\ref{subpower_low}) is a concave maximization problem.
\subsection{Overall Algorithm}
Based on the above three subsections, we propose JMS-T-P to solve problem (\ref{overall}). Specifically, in each iteration, we first optimize the communication mode $\bm{\beta}$ and subchannel allocation $\bm{A}$ by solving problem (\ref{cmsa}), given the transmit power $(\bm{P_U},\bm{P_M})$ and trajectory $\bm{q_U}$. Then, given the communication mode $\bm{\beta}$, subchannel allocation $\bm{A}$, and transmit power $(\bm{P_U},\bm{P_M})$, the trajectory of UAV $\bm{q_U}$ is optimized by solving problem (\ref{subtra}). Finally, transmit power $(\bm{P_U},\bm{P_M})$ is optimized by solving problem (\ref{subpower}), given communication mode $\bm{\beta}$, subchannel allocation $\bm{A}$ and trajectory $\bm{q_U}$. The JMS-T-P algorithm is summarized in Algorithm~\ref{alg_overall}\footnote{When applying SCP, the transmit power or the UAV trajectory is initialized with its most recently updated value.}.
\begin{algorithm}[!tpb]
	\caption{JMS-T-P algorithm}
	\label{alg_overall}
	\In{$\bm{P_U}^0,\bm{P_M}^0,\bm{q_U}^0$; Let $r=1$.}\\
	\Repeat{
		The increase of the objective value is below the pre-determined threshold $\epsilon>0$
	}
	{
		Solve problem (\ref{cmsa}) given $\{\bm{P_U}^{r-1},\bm{P_M}^{r-1},\bm{q_U}^{r-1}\}$, and denote the solution as $(\bm{\beta}^{r},\bm{A}^{r})$.\\
		Solve problem (\ref{subtra}) given $\{\bm{\beta}^{r},\bm{A}^{r},\bm{P_U}^{r-1},\bm{P_M}^{r-1}\}$, and denote the solution as $\bm{q_U}^{r}$.\\
		Solve problem (\ref{subpower}) given $\{\bm{\beta}^{r},\bm{A}^{r},\bm{q_U}^r\}$, and denote the solution as $(\bm{P_U}^{r},\bm{P_M}^{r})$.\\
		Update $r=r+1$.
	}
\end{algorithm}
\vspace{1cm}
\section{Overall Convergence and Complexity Analysis}
\label{subsection_scoc}
In this section, we first discuss the convergence of algorithm MSMA and JMS-T-P. Then, we provide the complexity of MSMA. Finally, some results on the system performance is given.

%Besides, since  $h_{U,B}^{k,t}$ and $h_{n,U}^{k,t}$ are lower bounded by $H_{U,B}^{k,t}$ and $H_{n,U}^{k,t}$, respectively, any feasible solution for the approximated horizontal position problem is also feasible for problem (\ref{subtra_hor}). 

\subsection{Convergence Analysis}
\begin{proposition}
	MSMA converges after a limited number of swap operations.
\end{proposition}
\begin{proof}
	Given swap operation $l$, we define the matchings before and after swap operation $l$ by $\Psi_{l-1}$ and $\Psi_l$, respectively. Besides, it is assumed that the pair of subchannels and their matches involved in this swap operation are $(k_1,k_2)$ and $(n_1,n_2)$, respectively, where $n_1=\Psi_{l-1}(k_1)$ and $n_2=\Psi_{l-1}(k_2)$. According to the definition of approved swap matching, after swap operation $l$, the utilities of UE $n_1$ and $n_2$ do not decrease and at least one of them increases. Besides, for any other UE, the matching of it under $\Psi_{l}$ is the same as that under $\Psi_{l-1}$, and thus, its utility keeps unchanged. Therefore, the total utility of the system strictly increases after one swap operation, which implies the matching after a swap operation is different from any matching generated before. Note that the number of matchings is limited. Therefore, we can conclude that MSMA converges after a limited number of swap operations.
\end{proof}

\begin{proposition}
	\label{J_converge}
	The JMS-T-P algorithm is guaranteed to converge.
\end{proposition}
\begin{proof}
	For simplicity of discussion, we define $g(\bm{\beta},\bm{A},\bm{P_U},\bm{P_M},\bm{q_U})=\sum\limits_{n\in\mathcal{N}}w_nR_n$. First, in iteration $r$ of JMS-T-P, the mode selection and subchannel allocation subproblem (\ref{cmsa}) is solved given $\{\bm{P_U}^{r-1},\bm{P_M}^{r-1},\bm{q_U}^{r-1}\}$. Therefore, we have
	\begin{align}
	\label{update_BA}
	g(\bm{\beta}^r,\bm{A}^r,\bm{P_U}^{r-1},\bm{P_M}^{r-1},\bm{q_U}^{r-1})\ge g(\bm{\beta}^{r-1},\bm{A}^{r-1},\bm{P_U}^{r-1},\bm{P_M}^{r-1},\bm{q_U}^{r-1}).
	\end{align}
	Then, we optimize the trajectory of the UAV relay, given $(\bm{\beta}^r,\bm{A}^r,\bm{P_U}^{r-1},\bm{P_M}^{r-1})$. Thus, we can obtain the following inequality
	\begin{align}
	\label{update_Q}
	g(\bm{\beta}^r,\bm{A}^r,\bm{P_U}^{r-1},\bm{P_M}^{r-1},\bm{q_U}^{r})\ge g(\bm{\beta}^r,\bm{A}^r,\bm{P_U}^{r-1},\bm{P_M}^{r-1},\bm{q_U}^{r-1}).
	\end{align}
	Afterwards, the power allocation $(\bm{P_U},\bm{P_M})$ is optimized, given $(\bm{\beta}^r,\bm{A}^r,\bm{q_U}^{r})$. As a result, we have
	\begin{align}
	\label{update_power_v2}
	g(\bm{\beta}^r,\bm{A}^r,\bm{P_U}^r,\bm{P_M}^r,\bm{q_U}^{r})\ge g(\bm{\beta}^r,\bm{A}^r,\bm{P_U}^{r-1},\bm{P_M}^{r-1},\bm{q_U}^{r}).
	\end{align}
	Based on (\ref{update_BA}), (\ref{update_Q}) and (\ref{update_power_v2}), we can conclude that
	\begin{align}
	g(\bm{\beta}^r,\bm{A}^r,\bm{P_U}^r,\bm{P_M}^r,\bm{q_U}^{r})\ge g(\bm{\beta}^{r-1},\bm{A}^{r-1},\bm{P_U}^{r-1},\bm{P_M}^{r-1},\bm{q_U}^{r-1}),
	\end{align}
	which means that the objective value of problem (\ref{overall}) is non-decreasing after each iteration of JMS-T-P algorithm. Besides, since the radio resource in the network is limited, there is an upper bound for the data rate of each UE. Thus, the objective value of problem (\ref{overall}) is upper bounded. Therefore, JMS-T-P is guaranteed to converge~\cite{HLY}.
\end{proof}
%\begin{proposition}
%	\label{TO_converge}
%	The TO algorithm will converge after a limited number of iterations. 
%\end{proposition}
%\begin{proof}
%	This proposition can be proved in the similar approach to that in \it{Proposition}~\ref{J_converge}. \rm{Specifically, the TO algorithm solves the horizontal position optimization and altitude optimization subproblems in an iterative way. After each iteration, the objective value of the trajectory optimization problem is non-decreasing. Since the objective value is upper bounded. Therefore, TO algorithm is guaranteed to converge after a limited number of iterations.} 
%\end{proof}
\subsection{Complexity Analysis}
\begin{proposition}
	In each iteration of MSMA, at most $\frac{1}{2}K(K-1)(N-1)$ swap matchings need to be considered. Given the number of iterations $I$, the complexity of MSMA is $O(INK^2)$
\end{proposition}
\begin{proof}
	Suppose subchannel $k_1$ is matched. For $k_1$, $N-1$ UE pair $(n_1,n_2)$ satisfying that $n_1$ is matched with $k_1$ and $n_2$ is not matched with $k_1$ can be found, because each subchannel is assigned to at most one UE and the number of UEs is $N$. For $n_2$, at most $K-1$ subchannels are matched with it. Therefore, a swap matching $\Psi_{k_1}^{k_2}$ with $k_1$ fixed has at most $(N-1)(K-1)$ combinations. Since the maximum number of matched subchannel is $K$, at most $\frac{1}{2}K(K-1)(N-1)$ swap matchings should be considered. Therefore, given the number of iterations $I$, the complexity of MSMA is $O(INK^2)$.
\end{proof}

\subsection{System Performance Analysis}
In the following proposition, we will discuss the trend of the average UAV speed when the maximum velocity of the UAV changes.
\begin{proposition}
	\label{pro_avg_speed} 
	When maximum velocity $v_{max}$ of the UAV increases, the average speed of the UAV becomes larger first, and then converges.
\end{proposition}
\begin{proof}
	We first show that when $v_{max}$ is small, the average speed of the UAV is positively correlated with $v_{max}$. To satisfy the QoS requirement, the UAV has to fly to different locations when serving different UEs. Define $l_{avg}$ as the average distance that the UAV flies when it turns to serve another UE. The average speed of the UAV relay can thus be expressed as $\frac{N_r l_{avg}}{TD_T}$, where $N_r$ is the number of UEs served by the UAV relay, and $D_T$ is the length of each time slot. Since the UE locations, number of time slots $T$, and slot length $D_T$ are given, we will focus on the influence of $v_{max}$ on $N_r$ in the following.
	
	To derive the expression of $N_r$, we first model the relationship between the sum of logarithmic data rate and $N_r$. The number of time slots that the UAV relay spends on changing locations to serve different UEs can be calculated as $T_f=\frac{N_rl_{avg}}{v_{f}D_T}$, where $v_f$ is the average velocity of the UAV when it flies to serve another UE. Therefore, one UE in the relay mode can occupy $\frac{T-T_f}{N_r}$ time slots on average to transmit signals to the BS. Note that when $v_f$ increases, the time spent on flying can be reduced, and thus, each UE in the relay mode has more time to transmit. Therefore, $v_f$ should be maximized. According to the expression of the flying power consumption, i.e.,
	\begin{align}
	P_f(v)=P_0(1+\frac{3v^2}{U_{tip}^2})+P_i(\sqrt{1+\frac{v^4}{4v_0^4}}-\frac{v^2}{2v_0^2})^{\frac{1}{2}}+\frac{1}{2}d_0\rho sA v^3,
	\end{align}
	the flying power increases with the UAV speed on the whole. Therefore, the battery capacity constraint can be neglected for now since $v_{max}$ is small. As a result, $v_f$ can be maximized as $v_{max}$. Denote the instantaneous data rate of one UE in the relay mode by $R_r$. The average data rate of such a UE can thus be given by $R_r\frac{T-T_f}{N_r}\frac{1}{T}$. Therefore, the sum of logarithmic data rate over all UEs can be expressed as 
	\begin{align}
	\sum_{n\in\mathcal{N}}R_n^{avg}=N_c\log_2(R_{avg}^c)+N_r\log_2(R_r\frac{T-T_f}{N_r T}),
	\end{align}
	where $N_c$ is the number of UEs in the cellular mode, and $R_{avg}^c$ is the average rate of one UE in the cellular mode. Then, we obtain the optimal $N_r$ based on the expression of the sum logarithmic rate. Specifically, we first formulate a function $f(x)=N_c\log_2(R_{avg}^c)+x\log_2(R_r\frac{T-T_f}{x T})$. By analyzing its derivative, we can find that $f$ is monotonically decreasing in $[0,\frac{R_r}{x_0+c})$, and increasing in $[\frac{R_r}{x_0+c},\infty)$, where $x_0$ is the zero point of the function $g(x)=\frac{2}{-\ln 2}-\frac{c}{\ln 2}x^{-1}+\log_2 x$, and $c=\frac{R_rl_{avg}}{Tv_{max}}$. Therefore, the optimal number of UEs served by the UAV relay can be approximated as
	\begin{align}
	N_r=\lfloor\frac{R_r}{x_0+c}\rfloor.
	\end{align}
	It can be found that when $v_{max}$ increases, $c$ becomes smaller. Note that $g$ is monotonically increasing, and $g(x)$ is negatively correlated with $c$ when $x$ is given. Therefore, zero point $x_0$ decreases with $v_{max}$, and thus, $N_r$ increases with $v_{max}$. As a result, we can conclude that the average UAV speed is positively correlated with maximum speed $v_{max}$. 
	
	Then, we show that when $v_{max}$ is large, the average speed will converge. Since the number of UEs in the relay mode is limited by $N$, the average speed cannot exceed $\frac{Nl_{avg}}{TD_T}$. Besides, due to the flying energy constraint, the flying power consumption in each time slot is smaller than $E_u$, and thus, the velocity of the UAV in each time slot is limited. Therefore, when $v_{max}$ is large, the average speed will converge.
\end{proof}

Consider any two disjoint regions within the cell coverage, denoted by $\mathcal{A}$ and $\mathcal{B}$, respectively. The numbers of UEs within region $\mathcal{A}$ and in $\mathcal{B}$ are represented by $N_A$ and $N_B$, respectively, satisfying $N_A \ge N_B$. Define $T_A$ and $T_B$ as the time period that the UAV spends within region $\mathcal{A}$, and $\mathcal{B}$, respectively. In the following, we consider the influence of UE numbers $N_A$ and $N_B$ on flight time $T_A$ and $T_B$.

\begin{proposition}
	\label{further}
	The UAV spends more time within regions with more UEs, i.e., $T_A\ge T_B$. Besides, when $N_A$ or $N_B$ is given, $\frac{T_B}{T_A}$ increases with $\frac{N_B}{N_A}$.
\end{proposition}

\begin{proof}
	Denote the instantaneous data rate of one UE within region $s$ by $R_s$, where $s\in\{\mathcal{A},\mathcal{B}\}$. To prove the proposition, we first model the relationship between flight time $(T_A,T_B)$ and the instantaneous data rate. Specifically, to maximize user fairness among the UEs within the two regions, $T_A$ and $T_B$ should be set such that the average rate of these UEs are the same. Since the average rate of one UE within region $s$ can be expressed as $\frac{R_sT_s}{T}$, we have
	\begin{align}
	\frac{R_AT_A}{T}=\frac{R_BT_B}{T},
	\end{align}
	which can be simplified into
	\begin{align}
	\frac{T_A}{T_B}=\frac{R_B}{R_A}.
	\end{align}
	Then, we investigate how numbers $N_A$ and $N_B$ of UEs affect instantaneous data rate $R_A$ and $R_B$. During $T_A$, the UEs in the region $\mathcal{A}$ can transmit to the BS simultaneously through the UAV relay due to OFDMA. However, the UEs in the region $\mathcal{B}$ cannot transmit in the relay mode due to severe pathloss between the UAV and these UEs. Therefore, the average UAV transmit power allocated to one UE in  region $\mathcal{A}$ is $\frac{P_U^{max}}{N_A}$. Similarly, during $T_B$, the UAV transmit power allocated to one UE in region $\mathcal{B}$ can be expressed as $\frac{P_U^{max}}{N_B}$. Since the instantaneous data rate of one UE is positively correlated with the allocated UAV transmit power, we have $R_A \le R_B$ since $N_A \ge N_B$. Therefore, it can be found that $T_A \ge T_B$.
	
	Then, we show that when $N_s~(s\in\{\mathcal{A},\mathcal{B}\})$ is given, $\frac{T_B}{T_A}$ increases with $\frac{N_B}{N_A}$. Without loss of generality, we assume that $s=\mathcal{B}$. When $\frac{N_B}{N_A}$ increases, number $N_A$ of UEs decreases since $N_B$ is given. Therefore, instantaneous data rate $R_A$ becomes larger while $R_B$ remains unchanged. Therefore, $\frac{T_B}{T_A}$ increases with $\frac{N_B}{N_A}$.
\end{proof}

\vspace{-0.5cm}
\section{Simulation results}\label{SIM}
In this section, we evaluate the performance of the proposed JMS-T-P algorithm. The selection of simulation parameters is based on exisiting works~\cite{SHBL-2018}. In the simulation, the initial horizontal positions of the UAV and the UEs are uniformly distributed within a circle with the BS as center and $200~m$ as radius. Besides, the altitude of the UAV relay $z_U^1$ obeys a uniform distribution between $[100~m,200~m]$. Other simulation parameters are listed in Table~\ref{sim_par}. 
\begin{table}[!tpb]
	\small
	\centering
	\caption{\normalsize{Simulation Parameters}}
	\label{sim_par}
	\begin{tabular}{|l|l|}	
		\hline	
		\textbf{Parameters} & \textbf{Values} \\
		\hline\hline
		Number of UEs $N$ & 5\\
		\hline
		Number of subchannels $K$ & 10 \\
		\hline
		Number of time slots $T$ & 10 \\
		\hline
		Duration of one time slot $\Delta T$ & 1 s\\	
		\hline	
		Height of BS $H_B$ & 30 m\\
		\hline
		Maximum moving distance $d_{max}$ & 15 m\\ 
		\hline
		Maximum transmit power of the UAV relay $P_U^{max}$ & 0.3 W\\
		\hline
		SNR threshold $\gamma^{th}$ & 300\\
		\hline
		Pathloss exponent for terrestrial communications $\alpha$ & 4\\
		\hline
		Noise variance $\sigma^2$ & -96 dBm\\
		\hline
		Air-to-ground channel parameters $\eta_{LoS}$ & 1 dB\\
		\hline
		Air-to-ground channel parameters $\eta_{NLoS}$ & 20 dB\\
		\hline
		Air-to-ground channel parameters $a$ & 9.6\\
		\hline
		Air-to-ground channel parameters $b$ & 0.28\\
		\hline
		Frequency of subchannel $k$ $f^k$ & 1 GHz\\
		%		\hline
		%		??? $K_c$ & 10 \\
		\hline
		Stop criterion for JMS-T-P $\epsilon$ & 0.001\\
		\hline
		Stop criterion for TO $\epsilon_{TO}$ & 0.01\\
		\hline
		ICI power $|I_B^k|^2$ & -110 dBm\\
		\hline
		Profile drag coefficient $\delta$ & 0.012 \\
		\hline
		Blade angular velocity $\Omega$ & 300 radians/seconds\\
		\hline
	  	Rotor radius $R$ & 0.4 m\\
		\hline
		Tip speed of the rotor blade $U_{tip}$ & 120 m/s\\
		\hline
		Mean rotor induced velocity in hover $v_0$ & 4.03\\
		\hline
		Fuselage drag ratio $d_0$ & 0.6\\
		\hline
		Air density $\rho$ & 1.225 $kg/m^3$\\
		\hline
		Rotor solidity $s$ & 0.05\\
		\hline
		Rotor disc area $A$ & 0.503 $m^2$\\
		\hline
		Means aircraft weight $W$ & 20 N \\
		\hline
		Incremental correction factor to induced power $k$ & 0.1 \\
		\hline
	\end{tabular}
\end{table}
%Note that in the following, the data rate refers to $\sum\limits_{n \in \mathcal{N}}{\log (R^{avg}_n+1)}$ and the throughput refers to the sum-rate of the UEs, i.e., $\sum\limits_{n \in \mathcal{N}}{R^{avg}_n}$.

For comparison, the following algorithms are also performed: 
\begin{enumerate}[itemindent=0em, label=$\bullet$]
	\item \textbf{Random algorithm}: We determine the communication mode and allocate subchannels randomly while the trajectory optimization and the power allocation are performed in the same way as in the JMS-T-P algorithm.
%	The mode selection and subchannel allocation, and power allocation are performed in the same way as the JMS-T-P algorithm while the trajectory is optimized in a greedy manner. Specifically, we determine the trajectory slot by slot. In each slot, we maximize the weighted sum rate over the UEs by optimizing the location of the UAV relay, where the weight is negatively correlated with the data rate in the past.
	\item \textbf{Cellular scheme}: All the UEs transmit in the cellular mode, and the subchannel allocation and the transmit power of the UEs are jointly optimized to maximize the weighted sum rate, i.e., $\sum\limits_{n \in \mathcal{N}}{w_n R_n}$.
	
%	\item \textbf{Fly-and-hover (FH) scheme:} The UAV flies to a certain location and then hovers there. The mode selection and subchannel allocation scheme, and the power allocation scheme of the FH algorithm are the same as those of the JMS-T-P algorithm.
\end{enumerate}

\begin{figure}[!tpb]
	\centering
	\includegraphics[width=4in]{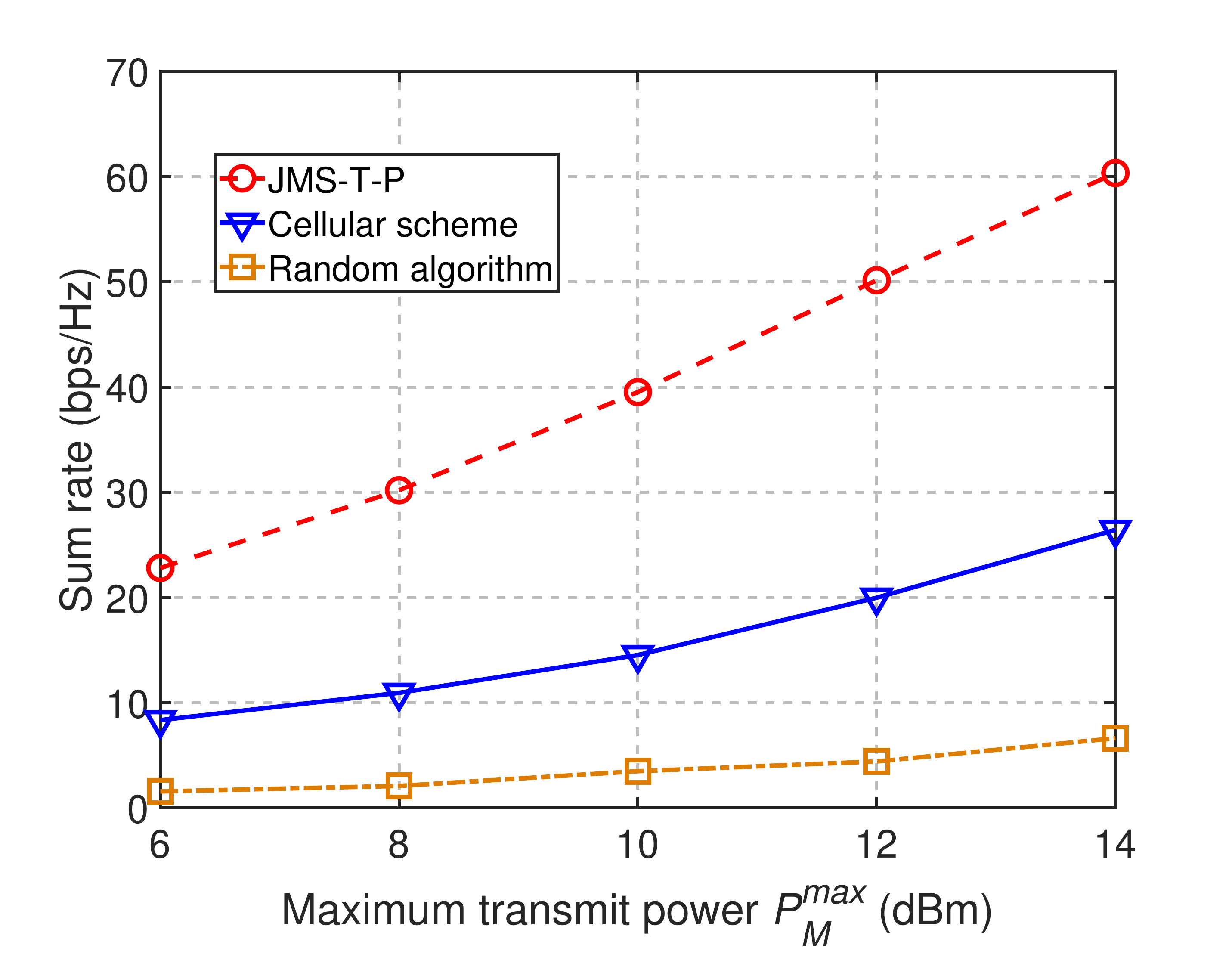}
	\caption{\small{Sum rate vs. maximum transmit power with $d_{max}=25$~m.}}
	\label{sum_rate_3cmp}
\end{figure}
Fig.~\ref{sum_rate_3cmp} depicts the influence of the maximum transmit power $P_M^{max}$ on the sum rate. From Fig.~\ref{sum_rate_3cmp}, we can find that the sum rate is positively correlated with the maximum transmit power $P_M^{max}$. Besides, it can be observed that the proposed JMS-T-P algorithm always outperforms the random algorithm in terms of the sum rate, since subchannel allocation and communication modes are determined randomly in the random algorithm. Moreover, Fig.~\ref{sum_rate_3cmp} shows that the JMS-T-P achieves larger sum rate than the cellular scheme, since in the cellular scheme, the UEs in the cell edge cannot transmit to the BS.   

\begin{figure}[!tpb]
	\centering
	\includegraphics[width=4in]{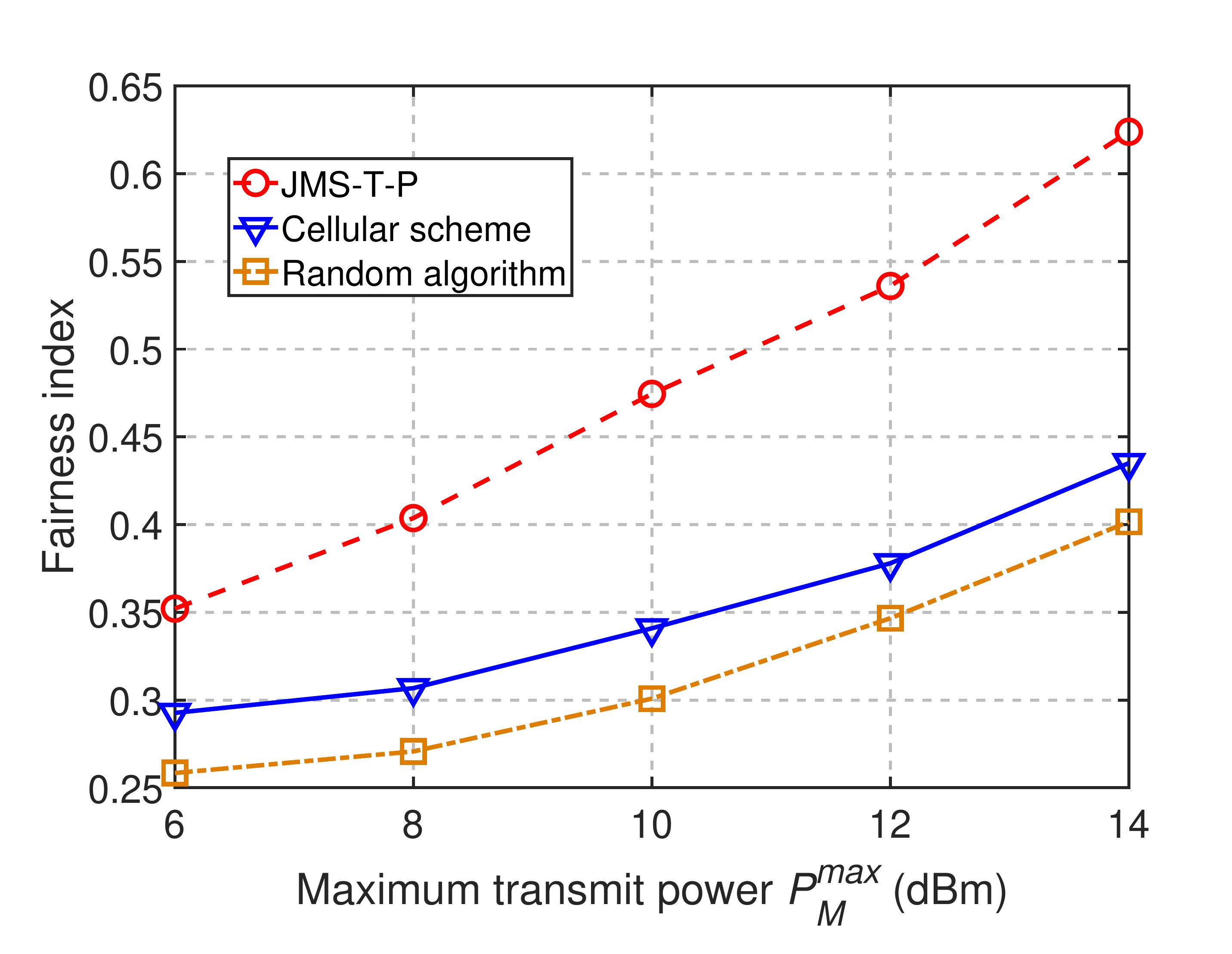}
	\caption{\small{Fairness vs. maximum transmit power with $d_{max}=25$~m}.}
	\label{fairness_3cmp}
\end{figure}
Fig.~\ref{fairness_3cmp} depicts the fairness index versus the maximum transmit power. To evaluate the user fairness, we introduce Jain's fairness index~\cite{RDW} which can be calculated as
\begin{align}
\frac{(\sum_{n \in \mathcal{N}}{R_n^{avg}})^2}{N\sum_{n \in \mathcal{N}}(R_n^{avg})^2}.
\end{align} 
The index varies from 0 to 1, and a larger value implies a higher fairness. As indicated in Fig.~\ref{fairness_3cmp}, higher transmit power will guarantee better user fairness, since more UEs can transmit to the BS. Besides, we can also find that the JMS-T-P algorithm achieves higher fairness than the cellular scheme, since the cell edge UEs cannot transmit to the BS in the cellular scheme. Moreover, according to Fig.~\ref{fairness_3cmp}, the JMS-T-P leads to better fairness among the UEs than the random algorithm, which demonstrates the effectiveness the JMS-T-P algorithm. 

\begin{figure}[!tpb]
	\centering
	\includegraphics[width=4in]{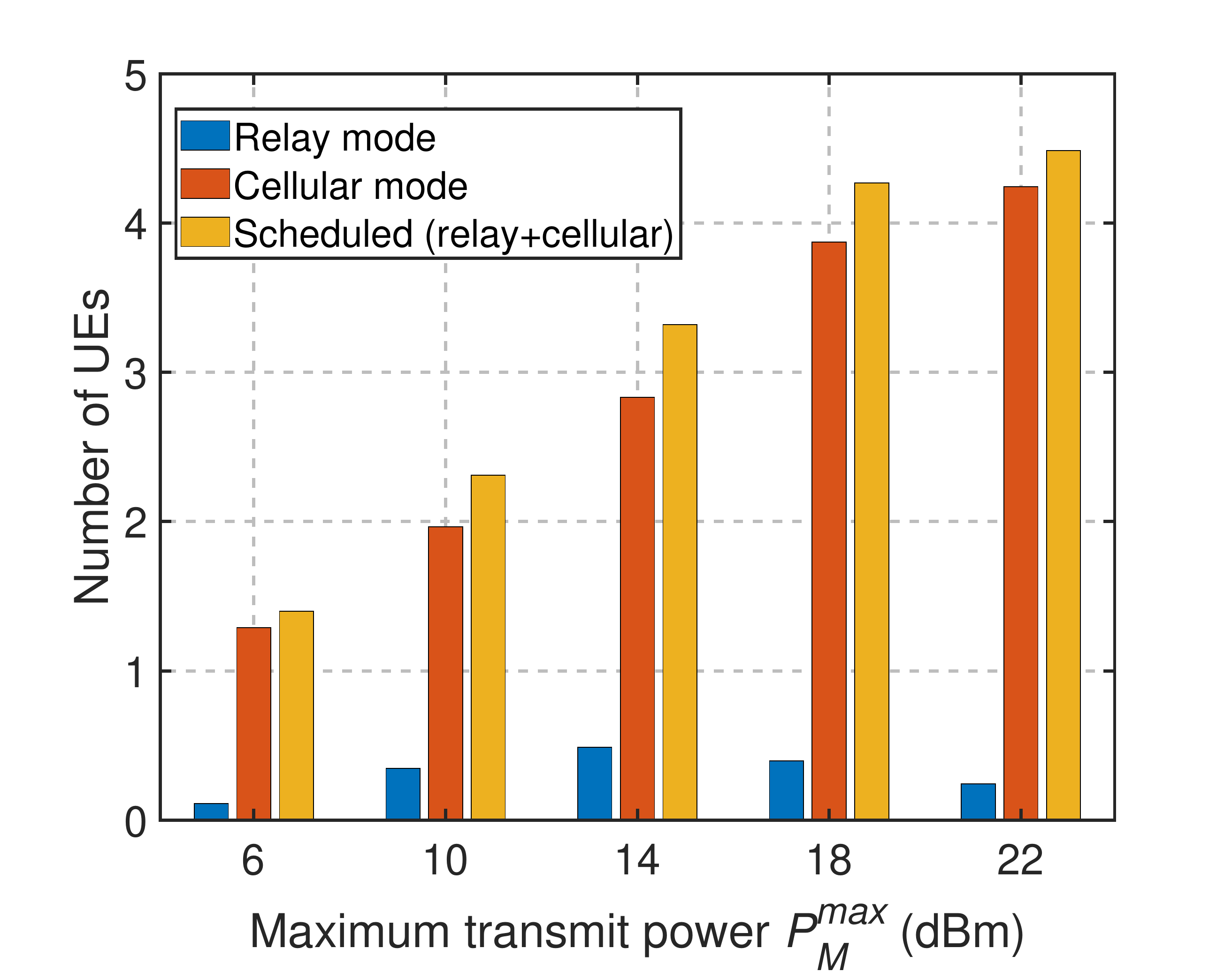}
	\caption{Number of UEs vs. maximum transmit power with $d_{max}=25$~m.}
	\label{bar}
\end{figure}
Fig.~\ref{bar} shows how the maximum transmit power $P_M^{max}$ influences the transmission modes of the UEs. We can observe that as the maximum transmit power increases, the number of scheduled UEs first increases and then saturates. Note that the SNR is positively correlated with the maximum transmit power. Therefore, when the maximum transmit power is small, more UEs will be scheduled if the maximum transmit power increases. However, the number of scheduled UEs will remain a constant when the maximum transmit power is sufficiently large since the number of UEs is fixed in this system. Besides, when the maximum transmit power is large, fewer UEs transmit in the relay mode as the maximum transmit power increases. This is because some UEs switch from the relay mode to the cellular mode.

\begin{figure}[!tpb]
	\centering
	\includegraphics[width=4in]{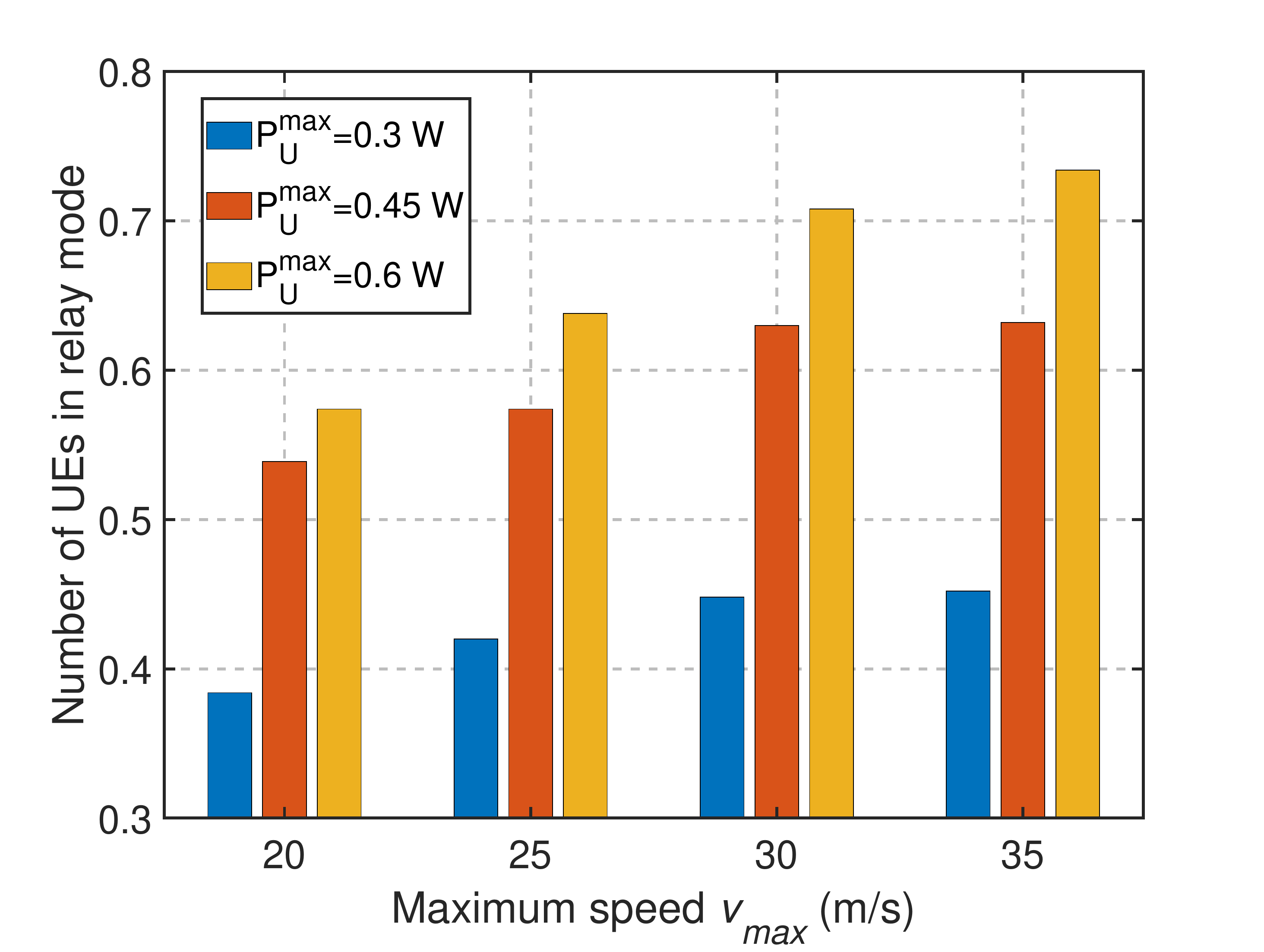}
	\scriptsize
	\caption{Number of UEs in the relay mode vs. maximum speed of the UAV relay, with $P_M^{max}=17$~dBm.}
	\label{ratio}
\end{figure}
Fig.~\ref{ratio} depicts the number of UEs in the relay mode versus the maximum speed of the UAV. It can be observed that when the mobility of the UAV relay strengthens, more UEs will transmit in the relay mode. This is because the UAV can spend less time  flying to serve another UE, during which no UE can transmit in the relay mode. Besides, we can also find that the number of UEs in the relay mode will decrease under a lower transmit power $P_U^{max}$. This is because the UAV relay has to move closer to the BS in order to satisfy the QoS constraint for the UAV-BS link, and thus, some UEs previously transmitting in the relay mode can not send data through the UAV relay now.

\begin{figure}[!tpb]
	\centering
	\includegraphics[width=3.8in]{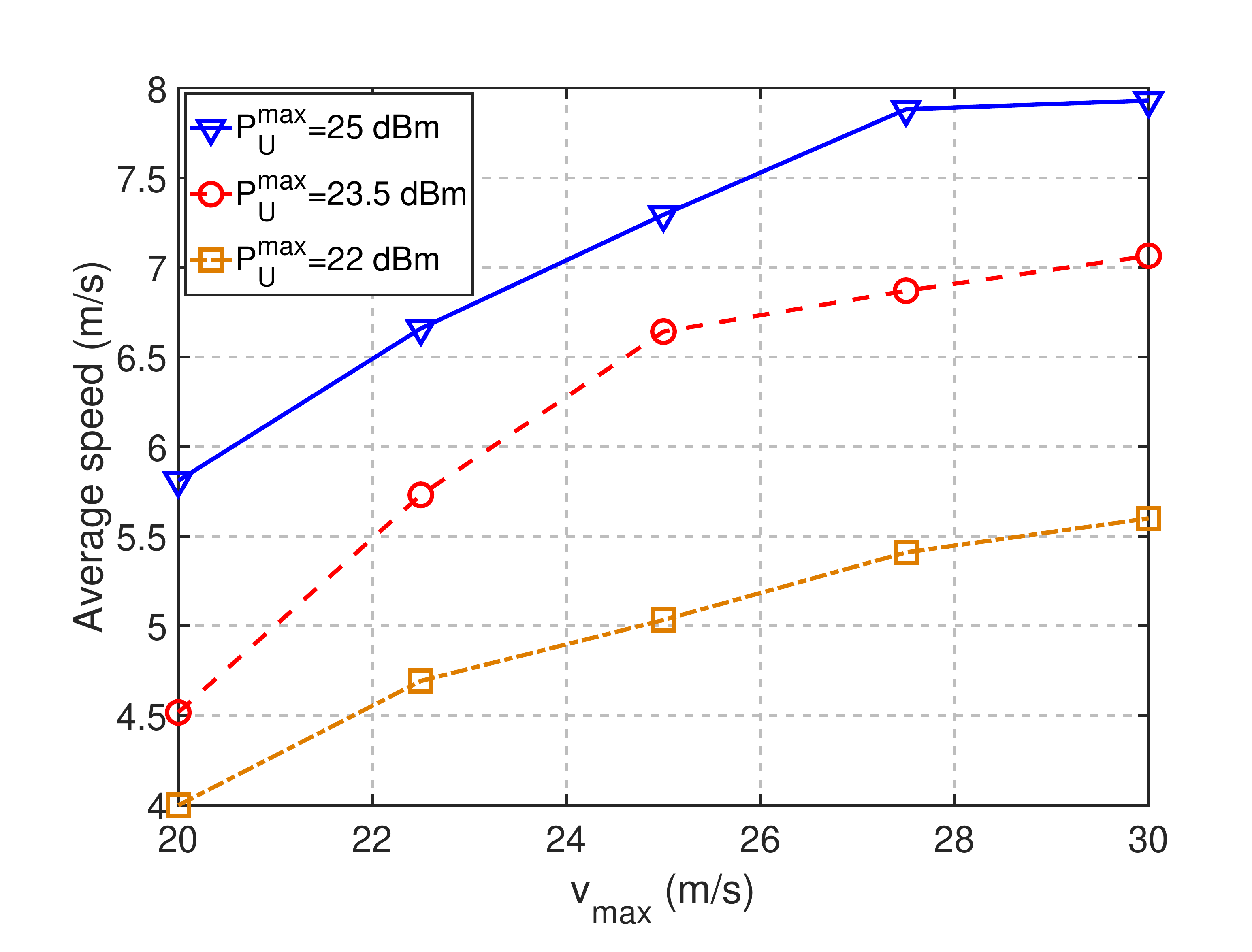}
	\caption{Average speed vs. maximum speed of the UAV relay, with $P_M^{max}=17$~dBm.}
	\label{avg_speed}
\end{figure} 

Fig.~\ref{avg_speed} depicts the average speed of the UAV relay versus maximum speed $v_{max}$. We can find that when $v_{max}$ increases, the average speed becomes larger first, and then converges, which is consistent with Proposition~\ref{pro_avg_speed}. Moreover, it can be observed that the average speed is positively correlated with the transmit power $P_U^{max}$. This is because there are more UEs which can transmit in the relay mode if the UAV relay flies to proper locations, and thus, the UAV moves for a larger distance in order to serve these additional UEs.

\vspace{-0.5cm}
\section{conclusion}\label{conclusion}
In this paper, we have studied a multi-user OFDMA cellular network with one UAV relay. Aiming to improve the throughput while guaranteeing the user fairness, we have formulated a series of joint mode selection, subchannel allocation, trajectory optimization and power allocation problems. Since the formulated problem is NP-hard, we have proposed the efficient JMS-T-P algorithm to obtain a suboptimal solution, where mode selection and subchannel allocation, trajectory optmization, and power allocation were performed iteratively. The simulation results have shown that JMS-T-P outperforms the random algorithm and the cellular scheme. Besides, based on the theoretical analysis, two important conclusions can be obtained. First, when maximum velocity of the UAV increases, the average speed of the UAV becomes larger first, and then converges. Second, the UAV spends more time within regions with more UEs.
%\begin{align}
%R_{n,R}^t=\sum_{k \in \mathcal{K}}{\frac{1}{2}\alpha_n^{k,t}\log_2(1+\frac{P_U^{k,t} P_n^{k,t} h^{k,t}_{U,B} h_{n,U}^{k,t}}{\sigma^2(P_n^{k,t} h_{n,U}^{k,t} +P_U^{k,t} h^{k,t}_{U,B} +\sigma^2)})}.\\
%R^t_{n,D}=\sum_{k \in \mathcal{K}}\alpha_n^{k,t}\log_2(1+\frac{P_n^{k,t} h^{k,t}_{n,B}}{\sigma^2}),\\
%R_n^t=\beta_n^tR_{n,R}^t+(1-\beta_n^t)R_{n,D}^t.\\
%\sum_{n \in \mathcal{N}}{\log (R^{avg}_n)}=\sum_{n \in \mathcal{N}}{\log\Big(\frac{1}{T}\sum_{t \in \mathcal{T}}{(\beta_n^tR_{n,R}^t+(1-\beta_n^t)R_{n,D}^t)}\Big)}
%\end{align}
\vspace{1cm}
\begin{appendices}
%	\section{Proof of Theorem \ref{clo_con}}
%	\label{app_clo_con}
%	First, we show that the feasible set of problem (\ref{subtra_hor}) is convex. Denote the feasible set by $\mathcal{F}$, and let $(\bm{x}_a,\bm{y}_a),(\bm{x}_b,\bm{y}_a) \in \mathcal{F}$ and $\alpha \in [0,1]$. In the following, we prove that the point $\bm{Q}_h=\alpha(\bm{x}_a,\bm{y}_a)+(1-\alpha)(\bm{x}_b,\bm{y}_a)$ belongs to the feasible set $\mathcal{F}$. Since the constraint function corresponding to the flying distance constraint is convex, the point $\bm{Q}_h$ satisfies constraint (\ref{}) 
%	According to the expression of the channel gain $h_{U,B}^{k}$, when the distance between the UAV and the BS $d_{U,B}$ and the elevation angle $\theta$ increase, $h_{U,B}^{k}$ becomes larger. Besides, 
	%%%%%%%%%%%%%%%%%%%%%%%%%%%%%%%%%%%%%%%%%below
	\section{Approximations for Inter-Subcarrier-Interference}
	\label{ICI}
	First, we focus on the ICI to subchannels allocated to UEs in the cellular mode. Specifically, the ICI to subchannel $k$ can be approximated by~\cite{PS_1999}
	\begin{align}
	\label{ICI_cellular}
	I_B^k=\sum_{n'\in\mathcal{N}_R}\sum_{k'\in\mathcal{K}_{n'}}\frac{\sin(\pi(f_{d,max}/\Delta f-k+k'))}{K\sin(\frac{\pi(f_{d,max}/\Delta f-k+k')}{K})}e^{j\pi(f_{d,max}/\Delta f-k+k')\frac{K-1}{K}}\sqrt{\frac{P_{n'}^{k'}G^{k'}K}{PL_{avg}^c}}h_{n',U}^{k'}X_{n'}^{k'},
	\end{align}
	where $\mathcal{N}_R$ is the set of UEs in the relay mode, $\mathcal{K}_{n'}$ represents the set of subcarriers allocated to UE $n'$, $K$ denotes the total number of subcarriers in the system, $f_{d,max}$ is maximum Doppler shift, $\Delta f$ denotes the subcarrier spacing, $PL_{avg}^c$ represents the average pathloss from the UAV to the BS over the center frequency, $G^{k'}$ is the amplification coefficient of the UAV relay over subchannel $k'$, and $X_{n'}^{k'}$ denotes the data symbol transmitted by UE $n'\in\mathcal{N}_R$ over subchannel $k'$ with unit power. Here, the maximum Doppler shift can be given by $f_{d,max}=\frac{vf_c}{c}$, where $v$ is the velocity of the UAV, $f_c$ represents the center frequency, and $c$ is the velocity of light in vacuum. By setting the parameters to their typical values\cite{SHBL-2018,36777}, the ratio of the ICI power to the desired signal power can be calculated to be $-16.1$~dB, which indicates that the ICI is weak compared with the desired signal. Here, the number of subchannels are set as $K=1000$, the subcarrier spacing is $\Delta f=15$~KHz, the center frequency is  $f_c=3.5$~GHz, the UAV speed is $v=100$~km/h, and the 	pathloss ratio is set to be $\eta=15$~dB.
	
	Then, we turn to the transmission in the relay mode, where an approximation for the ICI over subchannel $k$ is given below~\cite{PS_1999},
	\begin{align}
	\label{ICI_relay}
	I_B^k=&\sum_{n'\in\mathcal{N}_R, n'\neq n}\sum_{k'\in\mathcal{K}_{n'}}\frac{\sin(\pi(\epsilon-k+k'))}{K\sin(\frac{\pi(f_{d,max}/\Delta f-k+k')}{K})}e^{j\pi(f_{d,max}/\Delta f-k+k')\frac{K-1}{K}}\sqrt{\frac{P_{n'}^{k'}G^{k'}K}{PL_{avg}^c}}h_{n',U}^{k'}X_{n'}^{k'}\notag\\
	&+\sum_{k'\in\mathcal{K}_{n},k'\neq k}\frac{\sin(\pi(f_{d,max}/\Delta f-k+k'))}{K\sin(\frac{\pi(f_{d,max}/\Delta f-k+k')}{K})}e^{j\pi(f_{d,max}/\Delta f-k+k')\frac{K-1}{K}}\sqrt{\frac{P_{n}^{k'}G^{k'}K}{PL_{avg}^c}}h_{n,U}^{k'}X_{n}^{k'}.
	\end{align}
	Under the parameter settings, the ICI power can be found to be $-31.1$~dB lower than the desired signal power.
	 
	\section{Proof of Lemma \ref{lemma_sm}}
	\label{append_label_sm}
	We prove the first statement first. From $\sum_{k \in \mathcal{K}}{(\alpha_n^{k})^*R_{n,R}^{k}}>\sum_{k \in \mathcal{K}}{(\alpha_n^{k})^*R_{n,D}^{k}}$, we can infer that if UE $n$ adopts the relay mode, its average transmission rate $R^{avg}_n$ is larger than that when the cellular mode is adopted. Note that the communication mode $\beta_n$ will only affect $R_n$, which has a positive influence on weighted sum rate $\sum_{n \in \mathcal{N}}{w_nR_n}$. Therefore, sum rate $\sum_{n \in \mathcal{N}}{w_nR_n}$ under the relay mode is larger than that under the cellular mode. Besides, according to (\ref{SNR}), when UE $n$ transmits in either of the communication modes, the SNR requirements are satisfied. Thus, UE $n$
	will adopt the relay mode, i.e., $(\beta_n)^*=1$. Similarly, we can prove the second statement, and thus the proof is neglected here. 
%	\vspace{-1cm}
	\section{Proof of Theorem \ref{con_low_approx}}
	\label{Append_2}
	First, we prove that $H_{U,B}^{k}$ is a lower bound for $h_{U,B}^{k}$. It can be observed that the elevation angle $\theta$ in the formulation of $h_{U,B}^{k}$ is convex with respect to $\frac{d_{U,B}}{z_U-H_B}$. Recall that any convex
	function is globally lower-bounded by its first-order Taylor expansion at any point. Therefore, with point $C^{j-1}$, we have
	\vspace{-0.2cm}
	\begin{align}
	\label{low_theta}
	\theta \ge \frac{180}{\pi}\Big(\sin^{-1}(\frac{1}{C^{j-1}})+(\frac{-1}{C^{j-1}\sqrt{(C^{j-1})^2-1}})(\frac{d_{U,B}}{z_U^t-H_B}-C^{j-1})\Big)=\overline{\theta}.
	\end{align}
	Then, the probability of LoS connection from the UAV to the BS is lower bounded by,
	\vspace{-0.2cm}
	\begin{align}
	\label{Pr_low}
	PR_{LoS}\ge \frac{1}{1+a\exp\Big(-b(\overline{\theta}-a)\Big)}
	\end{align}
	Since $[1+a\exp\Big(-b(\overline{\theta}-a)\Big)]$ always takes positive values, the right-hand-side~(RHS) of (\ref{Pr_low}) is convex with respect to it. Therefore, with point $D^{j-1}$, we have 
	\begin{align}
	\label{Pr_low_1}
	\frac{1}{1+a\exp\Big(-b(\overline{\theta}-a)\Big)}
	\ge \frac{2}{D^{j-1}}-\frac{1}{(D^{j-1})^2}-\frac{a}{(D^{j-1})^2}\exp\Big(-b(\overline{\theta}-a)\Big)=\overline{PR_{LoS}},
	\end{align}
	which means $\overline{PR_{LoS}}$ is also a lower bound for $PR_{LoS}$. Then, by replacing $PR_{LoS}$ in the expression of $h_{U,B}^{k}$ with $\overline{PR_{LoS}}$, $h_{U,B}^{k}$ is lower bounded by
		\begin{align}
		\label{chan_gain_L}
		h_{U,B}^{k}\ge \frac{1}{(d_{U,B})^2\Big(L^k\eta_{NLoS}+(L^k\eta_{LoS}-L^k\eta_{NLoS})\overline{PR_{LoS}}\Big)}=\frac{1}{L^{k}}. 
		\end{align}
	Note that $\frac{1}{L^{k}}$ is convex with respect to $L^{k}$. Therefore, with point $L^{k}(x_U^{j-1},y_U^{j-1})$, we have
	\begin{align}
	\label{lower_bound}
	\frac{1}{L^{k}}\ge \frac{2}{L^{k}(x_U^{j-1},y_U^{j-1})}-\frac{1}{(L^{k}(x_U^{j-1},y_U^{j-1}))^2}L^{k}=H_{U,B}^{k}.
	\end{align}
	Based on (\ref{chan_gain_L}) and (\ref{lower_bound}), we can conclude that $H_{U,B}^{k}$ is a lower bound for $h_{U,B}^{k}$. Then, we show the concavity of $H_{U,B}^{k}$. It can be proved that $L^{k}$ is convex. Therefore, according to (\ref{lower_bound}), we can easily find that $H_{U,B}^{k}$ is a concave function with respect to $(x_U,y_U)$. 
	
	Similarly, we can prove that $H_{n,U}^{k}$ is a concave lower bound for $h_{n,U}^{k}$, and thus the proof is neglected here.
\vspace{-1cm}
\section{Proof of Theorem \ref{the_PP}}
\label{Append_3}
First, we show that $PR_{LoS}$ can be approximated by $\mathbb{PR}_{LoS}$. Since $z_U$ is feasible, the maximum moving distance of the UAV relay in the present time slot is $d_{max}$. Besides, the altitude $z_U$ is much larger than $d_{max}$. Therefore, the distance $d_{U,B}$ is approximately equal to $d_{U,B}(z_U)$, and thus the elevation angle $\theta$ in the expression of $PR_{LoS}$ can be rewritten as $\theta \approx \frac{180}{\pi}\sin^{-1}(\frac{z_U-H_B}{d_{U,B}(z_U)}+q_s)$,
%\begin{align}
%\theta^t \approx \frac{180}{\pi}\sin^{-1}(\frac{z_U^{t-1}-H_B}{d_{U,B}^t(z_U^{t-1})}+q_s).
%\end{align}
where $q_s=\frac{z_U-z_U}{d_{U,B}(z_U)}$. Note that $q_s$ is a tiny quantity. Therefore, $PR_{LoS}$ can be approximated by its first-order Taylor expansion with respect to $q_s$, which is calculated to be $\mathbb{PR}_{LoS}$. Similarly, we can prove that $\mathbb{PR}_{LoS}^{n}$ approximates $PR_{LoS}^{n}$, and thus the proof is neglected here.
\end{appendices}

%This problem can be transformed into
%\begin{subequations}
%\begin{align}
%	\label{main_subtra_1}
%	&\max_{\bm{Q},\lambda}{\lambda}\\
%	\label{cons_subtra1_1}
%	s.t. &\|\bm{q_U^{t+1}}-\bm{q_U^{t}}\| \le d_{max}, t=1, \dots, T-1,\\
%	\label{cons_subtra2_1}
%	&\sum_{n \in \mathcal{N}}{\log (R^{avg}_n)} \ge \lambda,
%\end{align}

%\begin{align}
%h_{n,U}^{k,t}=\frac{\|g_{U,n}^{k,t}\|^2}{(d_{n,U}^t)^2}\frac{1}{L^k \eta_{NLoS}+(L^k \eta_{LoS}-L^k \eta_{NLoS})Pr}
%\end{align}
%\end{subequations}

\end{document}